\pdfoutput=1
\documentclass[11pt]{article}

\usepackage[letterpaper,margin=1in]{geometry}

\usepackage{lipsum}

\usepackage{caption}

\usepackage[utf8]{inputenc} 
\usepackage[T1]{fontenc}    

\usepackage{url}            
\usepackage{booktabs}       
\usepackage{nicefrac}       
\usepackage{microtype}      
\usepackage{enumitem}
\usepackage{dsfont}
\usepackage{graphicx}
\usepackage{multicol}
\usepackage{xcolor}

\usepackage{amsthm,amsfonts,amsmath,amssymb,epsfig,color,float,graphicx,verbatim, enumitem, subfigure}

\usepackage{algorithm,algorithmicx}
\usepackage[ruled, vlined, algo2e]{algorithm2e}
\usepackage[noend]{algpseudocode}

\usepackage{mathtools}
\usepackage{bbm}

\usepackage{thm-restate}
\usepackage{turnstile}

\usepackage[colorlinks,citecolor=blue,linkcolor=magenta,bookmarks=true,hypertexnames=false]{hyperref}
\usepackage[nameinlink]{cleveref}

\crefname{equation}{equation}{equations}
\crefname{lemma}{lemma}{lemmata}
\crefname{claim}{claim}{claims}
\crefname{theorem}{theorem}{theorems}
\crefname{proposition}{proposition}{propositions}
\crefname{corollary}{corollary}{corollaries}
\crefname{claim}{claim}{claims}
\crefname{remark}{remark}{remarks}
\crefname{definition}{definition}{definitions}
\crefname{fact}{fact}{facts}
\crefname{question}{question}{questions}
\crefname{condition}{condition}{conditions}
\crefname{algorithm}{algorithm}{algorithms}
\crefname{assumption}{assumption}{assumptions}
\crefname{problem}{problem}{problems}

\newcommand{\jnote}[1]{\footnote{{\bf [Jongho: {#1}\bf ] }}}
\newcommand{\snote}[1]{\footnote{{\bf [Sushrut: {#1}\bf ] }}}

\newtheorem{theorem}{Theorem}[section]
\newtheorem{lemma}[theorem]{Lemma}

\newtheorem{corollary}[theorem]{Corollary}
\newtheorem{claim}[theorem]{Claim}

\newtheorem{definition}[theorem]{Definition}

\newtheorem{fact}[theorem]{Fact}

\theoremstyle{definition}

\newtheorem{remark}[theorem]{Remark}

\newcommand{\cons}{\text{\cons}}

\newcommand{\eps}{\epsilon}

\newcommand{\poly}{\mathrm{poly}}
\newcommand{\polylog}{\mathrm{polylog}}

\def\R{\mathbb R}

\def\Z{\mathbb Z}

\def\argmin{\mathrm{argmin}}
\newcommand{\citep}{\cite}
\newcommand{\citet}{\cite}

\newcommand{\cD}{\mathcal{D}}

\newcommand{\cL}{\mathcal{L}}

\newcommand{\cN}{\mathcal{N}}

\newcommand{\cW}{\mathcal{W}}

\newcommand{\paren}[1]{(#1)}
\newcommand{\Paren}[1]{\left(#1\right)}

\newcommand{\brac}[1]{[#1]}
\newcommand{\Brac}[1]{\left[#1\right]}

\newcommand{\abs}[1]{\lvert#1\rvert}

\newcommand{\Set}[1]{\left\{#1\right\}}

\newcommand{\norm}[1]{\lVert#1\rVert}

\newcommand{\relu}{\mathrm{ReLU}}
\newcommand{\sgn}{\mathrm{sign}}

\newcommand{\hide}[1]{}

\DeclareMathOperator*{\E}{\mathbf{E}}

\newcommand{\littlesum}{\mathop{\textstyle \sum}}

\let\vec\mathbf

\def\colorful{1}

\ifnum\colorful=0
\newcommand{\new}[1]{{\color{red} #1}}

\else
\newcommand{\new}[1]{{#1}}

\fi

\allowdisplaybreaks

\title{Distribution-Independent Regression for Generalized Linear Models with Oblivious Corruptions}

\author{
Ilias Diakonikolas\thanks{\new{Supported by NSF Medium Award CCF-2107079 and
a DARPA Learning with Less Labels (LwLL) grant.}}\\
University of Wisconsin-Madison\\
{\tt ilias@cs.wisc.edu}\\
\and
Sushrut Karmalkar\thanks{Supported by NSF under Grant \#2127309 to the Computing Research Association for the CIFellows 2021 Project.}\\
University of Wisconsin-Madison\\
{\tt skarmalkar@wisc.edu}\\
\and
Jongho Park\\
KRAFTON\\
{\tt jongho.park@wisc.edu}\\
\and
Christos Tzamos\thanks{Supported by NSF Award CCF-2008006 and NSF Award CCF-2144298 (CAREER).}\\
University of Wisconsin-Madison\\
{\tt tzamos@wisc.edu}\\
}

\begin{document}
	
	\maketitle
\begin{abstract}%
	We demonstrate the first algorithms for the problem of regression for generalized linear models (GLMs) in the presence of additive oblivious noise. 
	We assume we have sample access to examples $(x, y)$ where $y$ is a noisy measurement of $g(w^* \cdot x)$. In particular, \new{the noisy labels are of the form} $y = g(w^* \cdot x) + \xi + \eps$, where $\xi$ is the oblivious noise drawn independently of $x$ \new{and satisfies}  $\Pr[\xi = 0] \geq o(1)$, and $\eps \sim \cN(0, \sigma^2)$. Our goal is to accurately recover a \new{parameter vector $w$ such that the} function $g(w \cdot x)$ \new{has} arbitrarily small error when compared to the true values $g(w^* \cdot x)$, rather than the noisy measurements $y$. 
	
	We present an algorithm that tackles \new{this} problem in its most general distribution-independent setting, where the solution may not \new{even} be identifiable. \new{Our} algorithm returns \new{an accurate estimate of} the solution if it is identifiable, and otherwise returns a small list of candidates, one of which is close to the true solution. 
	Furthermore, we \new{provide} a necessary and sufficient condition for identifiability, which holds in broad settings. \new{Specifically,} the problem is identifiable when the quantile at which $\xi + \eps = 0$ is known, or when the family of hypotheses does not contain candidates that are nearly equal to a translated $g(w^* \cdot x) + A$ for some real number $A$, while also having large error when compared to $g(w^* \cdot x)$. 
	
	
	This is the first \new{algorithmic} result for GLM regression \new{with oblivious noise} which can handle more than half the samples being arbitrarily corrupted. Prior work focused largely on the setting of linear regression, and gave algorithms under restrictive assumptions.  
	
\end{abstract}
	
\setcounter{page}{0}

\thispagestyle{empty}

\newpage


\section{Introduction} \label{sec:intro}
Learning neural networks is a fundamental challenge in machine learning with various practical applications. 
Generalized Linear Models (GLMs) are the most fundamental building blocks of larger neural networks. 
These correspond to a linear function $w^* \cdot x$ composed with a {(typically non-linear)} 
activation function $g(\cdot)$. 
The problem of learning GLMs has received extensive attention in the past, 
especially for {the case of ReLU activations}. 
The simplest scenario is the ``realizable setting'', i.e., when the labels exactly match the target function, 
and can be solved efficiently with practical algorithms, 
such as gradient descent (see, e.g., \cite{Mahdi17}). 
In many real-world settings, noise comes from various sources, 
{ranging} from rare events and mistakes to skewed and corrupted measurements, 
making even simple regression problems {computationally} challenging.
In contrast to the realizable setting, when even a small amount of data is adversarially labeled, 
computational hardness results are known even 
for approximate recovery~\citep{HardtM13, MR18, DKMR22} 
and under well-behaved distributions~\citep{GoelKK19, DKZ20-sq-reg, GGK20, DKP21-SQ, DKR23}. 
To investigate more realistic noise models, \citet{chen2020online} and \citet{diakonikolas2021relu} study linear and $\relu$ regression in the Massart noise model, 
where an adversary has access to a \emph{random} subset of \emph{at most half} the samples 
and can perturb the labels arbitrarily after observing the uncorrupted samples.
By tackling regression in an intermediate {(``semi-random'') noise} model 
---  {lying} between the clean realizable and the adversarially labeled models --- 
these works recover $w^*$ under only mild assumptions on the distribution. 
{Interestingly, without any distributional assumptions, 
computational limitations have recently been established 
even in the Massart noise model~\citep{DK20-SQ-Massart, NasserT22, DKMR22-massart, DKRS22}.}

In this paper, we consider the problem of GLM regression under 
the {\em oblivious noise model} {(see Definition~\ref{def:glm reg oblivious noise})}, 
which is another intermediate model that allows the adversary to corrupt almost all the labels yet limits their capability by requiring the oblivious noise be determined independently of the samples.
The only assumption on this additive and independent noise is that it takes the value $0$ with \emph{vanishingly small} probability $\alpha>0$.
The oblivious noise model is a strong noise model
that {(information-theoretically)} allows for \emph{arbitrarily accurate} recovery of the target function. 
This stands in stark contrast to Massart noise, {where} it is impossible to recover the target function 
if more than half of the labels are corrupted. {On the other hand}, 
oblivious noise allows for recovery even when noise overwhelms, 
{i.e.,} as $\alpha \rightarrow 0$.

We formally define the problem of learning GLMs 
in the presence of additive oblivious noise below. 
As is the case with prior work on GLM regression (see, e.g., \cite{kakade2011efficient}), 
we make the standard assumptions that {the data distribution is supported in the unit ball} 
(i.e., $\norm{x}_2 \leq 1$) and that {the parameter space of weight vectors is bounded} 
(i.e, $\norm{w^*}_2 \leq R$).

\begin{definition}[GLM-Regression with Oblivious Noise]\label{def:glm reg oblivious noise}
	We say that $(x, y) \sim \text{GLM-Ob}(g, \sigma, w^*)$ if $x \in \R^d$ is drawn from some distribution supported in the unit ball 
	and $y = g(w^* \cdot x) + \xi + \eps$, 
	where $\eps$ and $\xi$ are drawn independent{ly} of $x$ 
	and satisfy $\Pr[\xi = 0] \geq \alpha = o(1)$ and $\eps \sim \cN(0, \sigma^2)$.
	We assume that $\norm{w^*}_2 \le R$ and that $g(\cdot)$ is $1$-Lipschitz and 
	monotonically {non-decreasing}. 
\end{definition}


In recent years, there has been increased focus on the problem of linear regression 
in the presence of oblivious noise~\citep{pesme2020online, dalalyan2019outlier, suggala2019adaptive, Tsakonas14, BhatiaJK15}. {This line of work has culminated in} consistent estimators when {the fraction of clean data is} $\alpha = d^{-c}$, 
where $c$ is a 
{small} constant~\citep{Steurer21Outliers}. 
In addition to linear regression, the oblivious noise {model} has also been studied 
for the problems of PCA, sparse recovery \citep{pesme2020online, d2021consistent}, 
and {in} the online setting~\citep{dalalyan2019outlier}. 
See \Cref{sec: prior work} for a detailed summary of related work.

However, prior algorithms and analyses often contain {somewhat} 
restrictive assumptions and exploit symmetry that only arises {for the special case} of linear functions. 
In this work, we address the following shortcomings of previous work: 
\begin{enumerate}[leftmargin=*]
	\item {\bf Assumptions on $\xi$ and marginal distribution}:
	Prior work either assumed that the oblivious noise was symmetric or 
	made strong distributional assumptions on the $x$'s, such as mean-zero Gaussian or sub-Gaussian tails. 
	We allow the distribution to be arbitrary (while being supported on the unit ball) 
	and make no additional assumptions on the oblivious noise. 
	\item {\bf Linear functions}: One useful technique to center an instance of the problem 
	for linear functions is to take pairwise differences of the data to induce symmetry. 
	This trick does not work for GLMs, since taking pairwise differences does not preserve the function 
	class we are trying to learn. Similarly, existing approaches do not generalize beyond linear functions. 
	Our algorithm works for a large variety of generalized models, including (but not restricted to) $\relu$s
	and sigmoids. 
\end{enumerate} 

{As our main result,} we demonstrate an efficient algorithm to efficiently recover $g(w^* \cdot x)$ 
if the distribution satisfies an efficient identifiability condition (see \Cref{def:identifiability}) 
and $\alpha = d^{-c}$ for any constant $c > 0$. 
If the condition {of  \Cref{def:identifiability}} does not hold, 
{our algorithm} returns a list of candidates,  
each of which is an approximate translation of $g(w^* \cdot x)$ 
and one of which is {guaranteed to be} as close to $g(w^* \cdot x)$ as we would like.  
In fact, if the condition does not hold, it is {information-theoretically} 
impossible to learn a unique function that explains the data.

\subsection{Our Results}\label{sec: our results}


{We start by noting that, }
at the level of generality we consider, 
the {learning} problem {we study} is not identifiable, 
i.e.,  multiple candidates in our hypothesis class might 
explain the data equally well. 
{As our first contribution}, we identify a necessary and sufficient condition characterizing when 
{a unique solution is identifiable}. 
We describe the efficient identifiability condition below. 

\begin{definition}[Efficient Unique Identifiability]\label{def:identifiability}
	We say $u$ and $v$ are $\Delta$-separated if  $$\E_x\Brac{\abs{g(u\cdot x) - g(v \cdot x)}} > \Delta.$$
	For any $\tau > 0$, an instance of the problem given in \Cref{def:glm reg oblivious noise} is $(\Delta ,\tau)$-identifiable if 
	any two $\Delta$-separated $u, v$ satisfy $\Pr_x \Brac{ \abs{ g(u \cdot x) - g(v \cdot x) - A} > \tau } > \tau$ for all $A\in \R$.  
\end{definition}

Let $\E_x\Brac{\abs{g(w \cdot x) - g(w^* \cdot x)}}$ denote the ``{excess} loss'' of $w$. 
Throughout the paper, we refer to $\Delta$ as the upper bound 
on the ``{excess} loss'' we would like to achieve. 
When {the problem is} $(\Delta, \tau)$-identifiable, 
the parameter $\tau$ describes the anti-concentration on the clean label difference $g(w \cdot x) - g(w^* \cdot x)$ centered around $A$.

Essentially, if there is a {weight vector} $w$ that is $\Delta$-separated from $w^*$, 
$(\Delta ,\tau)$-identifiability ensures that $g(w \cdot x)$ is not close to a translation of $g(w^* \cdot x)$. 
{On the other hand}, if $g(w \cdot x)$ is approximately a translation of $g(w^* \cdot x)$ for most $x$, the following lower bound shows that 
the adversary can design oblivious noise distributions so that $g(w \cdot x)$ and $g(w^* \cdot x)$ are indistinguishable.

\begin{theorem}[Necessity of Efficient Unique Identifiability]\label{prop:lower_bound}
Suppose that $\text{GLM-Ob}(g, \sigma, w^*)$ is not $(\Delta ,\tau)$-identifiable, 
i.e., there exist $u, v \in \R^d$ and $A \in \R$ such that $u, v$ are $\Delta$-separated 
and satisfy $\Pr_x \Brac{ \abs{ g(u \cdot x) - g(v \cdot x) - A} > \tau } \leq \tau$. 
Then any algorithm that distinguishes between $u$ and $v$ with probability at least $1-\delta$ 
requires $m=\Omega(\min(\sigma, 1) \ln(1/\delta) / \tau)$ samples.
\end{theorem}

{Note that} any algorithm that solves the oblivious regression problem 
must be able to differentiate between $w^*$ and any $\Delta$-separated candidate. 
\Cref{prop:lower_bound} explains the necessity of the efficient identifiability 
condition for such differentiation. 
If no $\tau > 0$ satisfies the condition, then \Cref{prop:lower_bound} implies that 
no algorithm {with finite sample complexity} 
can find a unique solution to oblivious regression. 
The result also shows that any $(\Delta, \tau)$-identifiable instance requires a sample complexity dependent on $1/\tau^*$, where the instance is $(\Delta, \tau)$-identifiable for all $\tau \le \tau^*$. 



Our main result is an efficient algorithm that performs GLM regression 
for any Lipschitz monotone activation function $g(\cdot)$. 
Our algorithm is qualitatively instance optimal -- whenever the problem instance $\text{GLM-Ob}(g, \sigma, w^*)$ is $(\Delta, \tau)$-identifiable, 
the algorithm returns a single candidate achieving 
{excess} loss of $4\Delta$ with respect to $g(w^* \cdot x)$. 
If not $(\Delta, \tau)$-identifiable, then our algorithm returns a list of candidates, 
one element of which achieves {excess} loss of $4\Delta$. 


\begin{theorem}[Main Result]\label{thm:main_informal}
There is an algorithm that given as input the desired accuracy $\Delta > 0$, 
an upper bound $R$ on $\norm{w^*}_2$, $\tau$, $\alpha$ and $\sigma$, 
it draws $m = \poly(d, R, \alpha^{-1}, \Delta^{-1}, \sigma)$ samples 
from $\text{GLM-Ob}(g, \sigma, w^*)$, 
it runs in time \new{$\poly(m, d)$}, 
and returns a $\poly(m)$-sized list of candidates, 
one of which achieves {excess} loss \new{at most} $\Delta$, 
i.e., there exists $\widehat w \in \R^d$ satisfying 
$\E_x \brac{ \abs{g(\widehat w \cdot x) - g(w^* \cdot x)} } \leq \Delta$. 

Moreover, if the problem instance is $(\Delta, \tau)$-identifiable 
(as in \Cref{def:identifiability}), 
then there is an algorithm which takes as input $\Delta, R, \alpha, \sigma$ and $\tau$, 
draws $\poly(d, R, \sigma, \alpha^{-1}, \Delta^{-1},\tau^{-1})$ samples, 
runs in time $\poly(d, R, \sigma, \alpha^{-1}, \Delta^{-1},\tau^{-1})$, 
and returns a single candidate. 
\end{theorem}

{ Our results hold for polynomially bounded $x$ and $w^*$ as well, by running the algorithm after scaling the $x$'s and reparameterizing $\Delta$. 
To see this, observe that we recover a $\widehat w$
such that $\mathbb E [|g(\widehat w \cdot x) - g(w^* \cdot x)|] \leq O(\Delta)$ for any choice of $\Delta$ when $\|x\|_2\leq 1$ and $\| w \|_2 \leq R$ for polynomially bounded $R$. 
Suppose instead of the setting for the theorem,
we have $\|x\|_2\leq A$ and $\|w\|_2\leq R$.
We can then divide the $x$'s by $A$ and interpret $y(x) = g(w \cdot x) =  g(Aw \cdot (x/A))$. We can then apply \Cref{thm:main_informal} with the upper bound on $w$ set to $AR$ and recover $\widehat w$, getting,  
$ \E [|g((\widehat w / A) \cdot x) - g(w^* \cdot x)|] 
 =  \E [|g(\widehat w \cdot (x/A)) - g(w^* \cdot x)|] \leq O(\Delta)$. }
Prior work on linear regression with oblivious noise either 
assumed that the oblivious noise was symmetric or 
that the mean of the underlying distribution was zero. 
Our result holds in a significantly more general setting, 
even for the special case of linear regression,
since we make no assumptions on the quantile of the oblivious noise 
or the mean of the underlying distribution. 

{At a high-level}, we prove \Cref{thm:main_informal} in three steps: 
(1) We create an oracle that, given a sufficiently close estimate of $\Pr[\xi \leq 0]$, 
generates a hyperplane that separates vectors achieving large loss 
with respect to $g(w^* \cdot x)$ from those achieving small loss, 
(2) We use online gradient descent to produce a list of candidate solutions, 
one of which is close to the actual solution, 
(3) We apply a unique tournament-style pruning procedure that eliminates 
all candidates far away from $w^*$. 

Since we do not have a good estimate of  $\Pr[\xi \leq 0]$, 
we run steps (1) and (2) for each candidate value of $1-2\Pr[\xi \leq 0]$ chosen from 
a uniform partition of $[-1, 1]$ and then perform (3) on the union of all these candidates.

\subsection{Technical Overview}

For simplicity of exposition, we will analyze the problem without additive Gaussian noise 
and when {the oblivious noise} $\xi$ is symmetric. 
This is the typical scenario for linear regression with oblivious noise 
in the context of general distributions. 
Inspired by the fact that the median of a dataset can be expressed as 
the $\ell_1$ minimizer
of the dataset, a natural idea is to minimize 
the $\ell_1$ loss $\cL_g(w) := \frac{1}{m} \sum_{i=1}^m \abs{y_i - g(x_i \cdot w)}$. 
This {simple approach} has been used in the context of linear regression with oblivious noise~\citep{NTN11} 
and also $\relu$-regression for Massart noise~\citep{diakonikolas2021relu}. 
Unfortunately, {if} the activation function $g$ is not linear, 
{the loss} $\cL_g(w)$ is not convex. 
Let $\cL^*_g(w) := \frac{1}{m} \sum_{i=1}^m \abs{g(w^* \cdot x_i) - g(w \cdot x_i)}$ denote the clean loss. 
To solve the problem of optimizing a nonconvex function, instead of using gradient-based methods, 
we can create an oracle that produces a separating hyperplane between points achieving a large clean loss and those achieving a small clean loss. 
The oracle produces a vector $H(w)$ satisfying $H(w) \cdot (w - w^*) \geq c > 0$. We then reduce the problem to online convex optimization (OCO). 

\paragraph{Oracle for Separating Hyperplane}

Unfortunately, unlike the case of convex functions --- 
or as it was used in  \citet{diakonikolas2021relu} to perform $\relu$ regression --- 
we cannot use $\nabla \cL_g(w) = \E_{x, y}[\sgn(g(w \cdot x) - y) \mathbf{1}(w \cdot x \geq 0)~x]$ 
as an oracle for generating a separating hyperplane, since it cannot distinguish $w = 0$ from $w^*$ even when $w^* \neq 0$. This is illustrated in \Cref{fig:comparison_standard_vs_ours} for $g = \relu$.
\begin{figure}
	\hfill
	\subfigure[Hyperplane based on $\nabla \cL_{\relu}(w)$]{\includegraphics[width=6cm]{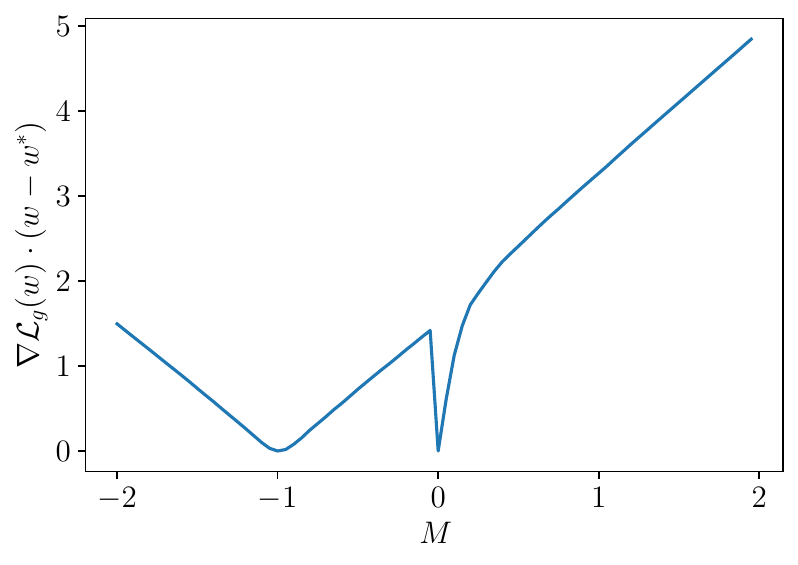}}
	\hfill
	\subfigure[Our separating hyperplane $H(w)$]{\includegraphics[width=6cm]{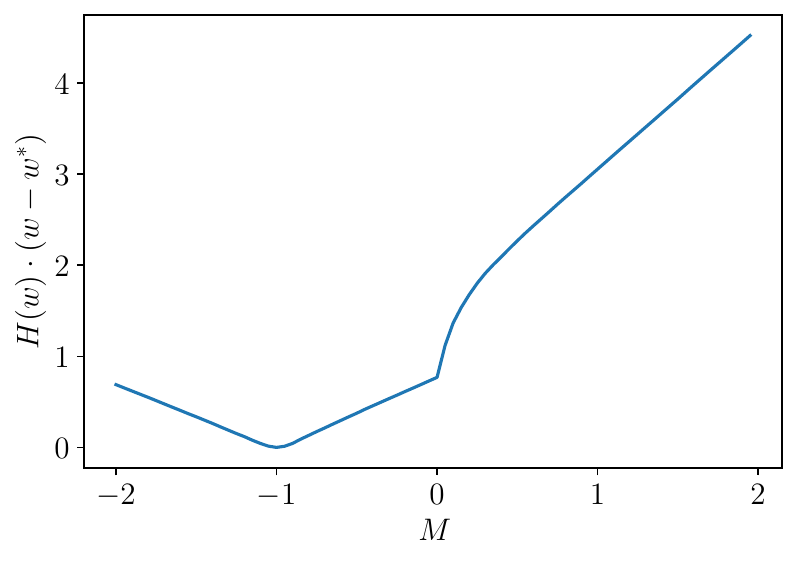}}
	\hfill
	\caption{We set $g = \relu$ and $w^* = -1$ and plot according to the line $w = -Mw^*$. Here, $H(w) :=  \E_{x, y} \Brac{\sgn(g(w \cdot x) - y) ~x}$. Observe that $\nabla \cL_{\relu}(w) \cdot (w - w^*) \rightarrow 0$ as $M\rightarrow 0^+$ even when $w^* \neq 0$. On the other hand, $H(w)$ does not suffer from this. 
	Here $\xi$ takes the value $-3/20$ w.p. $0.3$, $0$ w.p. $0.1$, and $3/20$ w.p. $0.6$; $\eps \sim \cN(0, 0.25)$, and $x$ is drawn uniformly from the $2$-dimensional unit ball. 
	}
		\label{fig:comparison_standard_vs_ours}%
\end{figure}

%
%

We instead take inspiration from the gradient of a \emph{linear regression} problem. 
Suppose we are given samples $(x_i, y_i)$ such that $y_i = w^* \cdot x_i + \xi'$, 
where $\xi'$ is symmetric oblivious noise such that $\Pr[\xi' = 0] \geq \alpha$, 
and the goal is to recover $\widehat{w}$ which is the $\ell_1$ minimizer, 
i.e., $\widehat w  := \argmin_{w \in \mathbb R^d} ~\cL'(w) := \frac{1}{m} \sum_{i=1}^m \abs{y_i - w\cdot x_i}.$
This is a convex program, and a subgradient of $\cL'(w)$ is given by 
$\nabla \cL'(w) =  \E_{x, y}[\mathrm{sign}(w \cdot x - y)~x] = \E_{x}\Brac{\E_{\xi'} \Brac{\sgn((w\cdot x - w^* \cdot x) - \xi')}x}.$

We now examine the expectation over $\xi'$. 
Since the median of $\xi'$ is $0$ and it takes the value $0$ with probability at least $\alpha$, 
the probability that $\sgn((w\cdot x - w^* \cdot x)-\xi') = \sgn(w\cdot x - w^* \cdot x)$ 
is at least $\frac{1+\alpha}{2}$. 
This implies that 
$\E_{\xi'} \Brac{\sgn((w\cdot x - w^* \cdot x) - \xi')} \geq \alpha ~\sgn(w\cdot x - w^* \cdot x)$, 
since $(w\cdot x - w^* \cdot x) - \xi'$ is more often biased towards 
$(w\cdot x - w^* \cdot x)$ than it is towards $-(w\cdot x - w^* \cdot x)$ and $\Pr[\xi = 0] \geq \alpha$. 
Therefore, $\nabla \cL'(w)  \cdot (w - w^*) \geq \alpha  \E_{x}[\abs{w\cdot x - w^* \cdot x}].$

While we do not have access to $w^* \cdot x + \xi'$ we do have access to $g(w^* \cdot x) + \xi$. 
At this point we make two observations:  
(1) Since $g$ is monotonically {non-decreasing}, it follows that 
$\sgn(g(w \cdot x) - g(w^* \cdot x)) = \sgn(w\cdot x - w^* \cdot x)$ whenever $g(w\cdot x) \neq g(w^* \cdot x)$.
(2) Since $g$ is $1$-Lipschitz, it follows that 
$\abs{w\cdot x - w^* \cdot x} \geq \abs{g(w\cdot x) - g(w^* \cdot x)}$. 
An argument analogous to the one above then shows us that $H(w) := \E_{x, y} \Brac{\sgn(g(w \cdot x) - y) ~x}$ satisfies 
$H(w) \cdot (w - w^*) \geq \alpha \E_x [\abs{g(w\cdot x) - g(w^* \cdot x)}]$, 
hence allowing us to separate $w$'s which achieve small clean loss 
from those which achieve larger clean loss. 
In \Cref{lem:GD_lower_bound}, we demonstrate this in the presence of additive Gaussian noise 
and without the assumption of symmetry on $\xi$.

\paragraph{Reduction to Online Convex Optimization} 
In \Cref{thm:OCO_reduction}, we show that if we have a good estimate 
of the quantile at which $\xi$ is $0$, 
we can use our separating hyperplane oracle as the gradient oracle 
for online gradient descent to optimize the clean loss 
$\cL^*_g(w) := \E_x \Brac{ \abs{g(w \cdot x) - g(w^* \cdot x)}}$. 
Since this function is nonconvex, our reduction leaves us with a set of candidates 
which are iterates of our online gradient descent procedure. 
{Our minimizer is one of these candidates}.  
We then prune out candidates which do not explain the data. 

\paragraph{Pruning Bad Candidates}
Finally, \Cref{lem:prune_bad_candidates} shows that we can efficiently 
prune implausible candidates if the list of candidates contains a vector close to $w^*$. 
For simplicity of exposition, assume for now that $w^* \in \cW$. 
Our pruning procedure relies on the following two observations: 
(1) There is no way {to} find disjoint subsets of the space of $x$'s 
such that $(y - g(w^* \cdot x))$ takes the value $0$ at different quantiles when conditioned on these subsets. 
(2)  {Suppose that}, for some $A$, we identify $E^+$ and $E^-$ such that $x \in E^+$ 
implies $g(w \cdot x) - g(w^* \cdot x) - A> \tau$ 
and $x \in E^-$ implies $ g(w \cdot x) - g(w^* \cdot x) - A< -\tau$.  
Then the quantiles at which $(y - g(w \cdot x)) = (g(w^* \cdot x) - g(w \cdot x) + \xi + \eps)$ 
take the value $0$ in these two sets differ by at least $\alpha$. 

We can use these {observations} to determine if a given candidate 
is equal to $w^*$ or not, by looking {at the quantity} $g(w \cdot x) - g(w^* \cdot x) - A$.
{Specifically,} we try to find two subsets $E^+$ and $E^-$
such that $g(w \cdot x) - g(w^* \cdot x)$ is large and positive when $x \in E^+$, 
and is large and negative when $x \in E^-$. 
We reject $w$ by comparing the quantiles of $(y - g(w \cdot x))$ 
when conditioned on $x$ belonging to $E^+$ and $E^-$. 
While we do not know what $w^*$ is beforehand, we know that $w^* \in \cW$, 
and so we iterate over elements in $\cW$ to check for the existence 
of a partition which will allow us to reject $w$. 
If $w = w^*$ such a partition is not possible, and $w$ will not be rejected. 
On the other hand, each candidate remaining in the list 
will be close to a translation of $g(w^* \cdot x)$, and one of the candidates will be $w^*$.

\subsection{Prior Work}\label{sec: prior work}

Given the extensive literature on robust regression, here we focus on the most relevant work. 

\paragraph{GLM regression} 
{Various formalizations of} GLM regression 
{have} been studied extensively {over the past decades}; 
see., e,g., \citet{nelder1972generalized, kalai2009isotron, kakade2011efficient, klivans2017learning}. 
Recently, there has been increased focus on GLM regression {for activation functions that are popular in deep learning, including $\relu$s}. 
This problem has previously been considered both in the context of weaker noise models, 
such as the realizable/\new{random additive noise} setting ~\citep{Mahdi17, kalan2019fitting,yehudai2020learning}, 
as well as for more challenging noise models, 
\new{including adversarial label noise}~\citep{GoelKK19, DKZ20-sq-reg, GGK20, DKP21-SQ, DGKKS20, DKMR22, DKTZ22, WZDD23}. 

Even in the realizable setting \new{(i.e., with clean labels)}, 
it turns out that the squared loss has exponentially many local minima 
for the logistic activation function~\citep{Auer95}.
On the positive side,~\citet{DGKKS20} gave an efficient learner 
\new{achieving a constant factor approximation} 
{in the presence of adversarial label noise} 
under isotropic logconcave distributions. 
{This algorithmic result was 
generalized to broader families of activations under much milder distributional assumptions in~\citet{DKTZ22, WZDD23}.}
On the other hand, without distributional assumptions, 
even approximate learning is \new{computationally} hard~\citep{HardtM13, MR18, DKMR22}. 
{In a related direction, the problem has been studied in 
the distribution-free setting under semi-random label noise.}
Specifically,~\citet{diakonikolas2021relu} focused on the setting 
of bounded (Massart) noise, where the adversary can arbitrarily corrupt 
a randomly selected subset of at most half the samples.
Even earlier,~\citet{karmakar2020study} studied the realizable setting 
under a noise model similar to (but more restrictive than) the Massart noise model, while~\citet{chen2020classification} studied a classification version of learning 
GLMs with Massart noise. 

In our work, we study the problem of \new{distribution-free learning of general GLMs 
in the presence of oblivious noise}, 
with the goal of being able to tolerate $1-o(1)$ fraction of the samples being corrupted.
In this setting, we recover the candidate solution 
to arbitrarily small precision in $\ell_1$ norm {(with respect to the objective}). 
Since $\norm{x}_2 \leq 1$ and $\norm{w}_2 \leq R$, it is easy to also provide the corresponding guarantees in the $\ell_2$ norm, which was the convention in \new{some earlier} works \citep{kakade2011efficient, GoelKK19}.

\new{\paragraph{Algorithmic Robust Statistics}
A long line of work, initiated in~\cite{DKKLMS16, LaiRV16}, has focused on designing robust estimators
for a range of learning tasks (both supervised and unsupervised)
in the presence of a small constant fraction of adversarial corruptions; see~\cite{DiaKan22-book} for a textbook overview of this field. In the supervised setting, the underlying contamination model allows for corruptions in both $x$ and $y$ and makes no further assumptions on the power of the adversary. 
A limitation of these results is the assumption that the good data 
comes from a well-behaved distribution. 
On the other hand, without distributional assumptions on the clean data, these problems become computationally intractable. This fact motivates the study of weaker --- but still realistic --- noise models in which efficient noise-tolerant algorithms are possible in the distribution-free setting.
}

\new{\paragraph{Adversarial Label Corruptions}

In addition to adversarial corruptions in both $x$ and $y$, the adversarial label corruption model is another model which has been extensively studied in the context of linear and polynomial regression problems 
\cite{laska2009exact, Li2011CompressedSA, BhatiaJK15, BhatiaJKK17, kkp17, KarmalkarP19}. 
These results make strong assumptions on the distribution over the marginals. In contrast, we make no assumptions beyond the marginals coming from a bounded norm distribution.}

\paragraph{The Oblivious Noise Model} 
The oblivious noise model could be viewed as an attempt 
at characterizing the most general noise model 
that allows almost all points to be arbitrarily corrupted, 
while still allowing for recovery of the target function with vanishing error. 
This model has been considered for natural statistical problems, 
including PCA, sparse recovery~\citep{pesme2020online, d2021consistent}, 
as well as linear regression in the online setting~\citep{dalalyan2019outlier} and the problem of estimating a signal $x^*$ with additive oblivious noise~\citep{SoSOblivious21}.

The setting closest to the one considered in this paper is that of linear regression. Until very recently, the problem had
been studied primarily in the context of Gaussian design matrices, 
i.e., when $x$'s are drawn from $\cN(0, \Sigma)$. 
One of the main goals in this line of work is to design an algorithm 
that can tolerate the largest possible fraction of the labels being corrupted.
Initial works on linear regression either were not consistent as the error did 
not go to $0$ with increasing samples ~\citep{wright2008dense, NTN11} 
or failed to achieve the right convergence rates or breakdown point~\citep{Tsakonas14, BhatiaJK15}. \citet{suggala2019adaptive} provided the first consistent estimator 
achieving an error of $O(d/\alpha^2 m)$ for any $\alpha > 1/\log \log(m)$. 
Later, \citet{Steurer21Outliers} improved this rate to $\alpha > 1/d^c$ for constant $c$, while also generalizing the class of design matrices. 

Most of these prior results focused on either the oblivious noise 
being symmetric (or median $0$), or the underlying distribution being {(sub)-}Gaussian. 
In some of these settings (such as that of linear regression) 
it is possible to reduce the general problem 
to this restrictive setting, as is done in~\citet{norman2022robust}. 
However, for GLM regression, we cannot exploit the symmetry that is either induced by the distribution or the class of linear functions. 
In terms of lower bounds,~\citet{HT22} identify a ``well-spreadness'' condition (the column space of the measurements being far from sparse vectors) 
as a property that is necessary for recovery even when the oblivious noise is symmetric. 
Notably, these lower bounds are relevant when the goal is to perform parameter recovery or achieve a rate better than $\sigma^2/m$. 
In our paper, we instead give the first result for a far more general problem and with the objective of minimizing the clean loss, but not necessarily parameter recovery. 
Our lower bound follows from the fact that 
we cannot distinguish between translations of the target function from the data without making any assumptions on the oblivious noise.

	\section{Preliminaries}

\paragraph{Basic Notation} 
We use $\R$ to denote the set of real numbers.
For $n \in \Z_+$ we  denote $[n] := \{1,\ldots,n\}$. 
 We assume $\sgn(0) = 0$. 
We denote by $\mathbf{1}(E)$ the indicator function of the event $E$. 
{We use $\poly(\cdot)$ to indicate a quantity that is polynomial in its arguments. 
	Similarly, $\polylog(\cdot)$ denotes a quantity that is polynomial in the logarithm of its arguments. 
For two functions $f, g : \R \rightarrow \R$, we say $f \lesssim g$ if there exist constants $C_1, C_2 > 0$ such that for all $x \geq C_1$, $f(x) \leq C_2 g(x)$. 
For two numbers $a, b \in \R$, $\min(a, b)$ returns the smaller of the two. 
We say that a function $f$ is $L$-Lipschitz if $f(x) -f(y) \le L \|x-y\|_2$.

\paragraph{Linear Algebra Notation} 
We typically use small case letters for deterministic vectors and scalars. 
For a vector $v$, we let $\|v\|_2$ denote its $\ell_2$-norm. 
We denote the inner product of two vectors $u, v$ by $u \cdot v$. 
We denote the $d$-dimensional radius-$R$ ball centered at the origin by $B_d(R)$. 

\paragraph{Probability Notation} 
For a random variable $X$, we use $\E[X]$ for its expectation and $\Pr[X \in E]$ for the probability of the random variable belonging to the set $E$.
We use $\cN(\mu,\sigma^2)$ to denote the Gaussian distribution 
with mean $\mu$ and variance $\sigma^2$. 
When $D$ is a distribution, we use $X \sim D$ to denote that 
the random variable $X$ is distributed according to $D$. 
When $S$ is a set, we let $\E_{X \sim S}[\cdot]$ denote the expectation
under the uniform distribution over $S$. When clear from context, we denote the empirical expectation and probability by $\widehat \E$ and $\widehat \Pr$.

\paragraph{\new{Basic Technical} Facts}

The proofs of the following facts can be found in \Cref{app:facts}. 

\begin{fact}\label{fact:expectation_eps_xi}
	Let $\xi$ be oblivious noise such that $\Pr[\xi = 0] \ge \alpha$.  
	Then the quantity $$F_{\sigma, \xi}(t) := \E_{\eps, \xi}[\sgn(t + \eps + \xi)] -  \E_{\eps, \xi}[\sgn(\eps + \xi)]$$ 
	satisfies the following: 
	(1) $F_{\sigma, \xi}$ is strictly increasing, 
	(2)  $\sgn(F_{\sigma, \xi}(t)) = \sgn(t)$, and 
	(3) For any $\gamma \leq 2$, whenever $|t| \geq \gamma \sigma $, $\abs{F_{\sigma, \xi}(t) } > (\gamma \alpha/4)$ and whenever $\abs{t} \leq \gamma \sigma$, $\abs{F_{\sigma, \xi}(t) } \leq (\alpha t/4 \sigma)$. 
\end{fact}

\begin{fact}\label{fact:small_tail_prob}
	Let $X$ be a random variable on $\R$. Fix $\tau > 0$ and $\eta > 0$. Define the events $E_A^+$ and $E_A^-$ such that
	$\Pr[E_A^+] = \Pr[X > A + \tau]$ and $\Pr[E_A^-] = \Pr[X < A - \tau]$. Then if the following first condition is not true, the second condition is:
	(1) $\exists A \in \R$ such that $\Pr[E_A^+] \geq \eta$ and $\Pr[E_A^-] \geq \eta$.
	(2) $\exists A^* \in \R$ such that $\Pr[E_{A^*}^+] \le \eta$ and $\Pr[E_{A^*}^-] \le \eta$.
\end{fact}

	\section{Oblivious Regression via Online Convex Optimization}\label{sec:oco}
	\subsection{A Direction of Improvement}
We assume prior knowledge of a constant $c$ that approximates $\E_{\xi, \eps}[\sgn(\xi + \eps)]$. 
In the following key lemma, we demonstrate an oracle for a hyperplane that separates all vectors that are $\Delta$-separated from $w^*$. 
For the following results in Section~\ref{sec:oco} and later in the paper,  we use $\gamma$ to denote $\min(\Delta/4\sigma, 1/2)$.

	\begin{lemma}[Separating Hyperplane]\label{lem:GD_lower_bound} 
	Let $\cD = \text{GLM-Ob}(g, \sigma, w^*)$ as defined in \Cref{def:glm reg oblivious noise} 
	and define $\gamma = \min(\Delta/4\sigma, 1/2)$. 
	Suppose $c \in \R$ such that $\abs{\E_{\xi, \eps}[\sgn(\xi + \eps)] - c} \leq \gamma \alpha \Delta / 32 R$.
	Then, for $w \in B_d(R)$, $H_c(w) := \E_{x, y}\Brac{\Paren{ \mathrm{sign}(y - g(w \cdot x)) - c}~x}$ satisfies
	\[H_c(w) \cdot (w^* - w) \geq (\gamma \alpha/4) \E_x\Brac{\abs{(g(w^*\cdot x) - g(w\cdot x))}} 
	- (\gamma^2 \alpha \sigma / 4)  - (\gamma \alpha \Delta / 16) .\]
	Specifically, if $\E_x[\abs{(g(w^* \cdot x) - g(w \cdot x))}] > \Delta$, 
	we have that $H_c(w) \cdot (w^* - w) \geq (\alpha \Delta^2)/(32 \sigma)$ if $\Delta \leq 2\sigma$, 
	and $H_c(w) \cdot (w^* - w) \geq \alpha \Delta / 8$ if $\Delta > 2\sigma$.
	\end{lemma}
	\begin{proof}
		Let $F_{\sigma, \xi}(t) := \E_{\eps, \xi}[\sgn(t + \eps + \xi) - \sgn(\eps + \xi)]$. Then we can write 
		\begin{align*}
			H_c(w) \cdot (w^* - w) &=\E_{x, \eps, \xi} \brac{\Paren{\sgn(g(w^*\cdot x) - g(w\cdot x) + \eps + \xi)-c}~(x \cdot (w^* - w) )} \\
			&=\E_{x} \Brac{\Paren{F_{\sigma, \xi} (g(w^*\cdot x) - g(w\cdot x)) + \paren{\E_{\xi, \eps }[\sgn(\xi +\eps)] - c}}~(x \cdot  (w^* - w) )} \\
			&=\E_{x} \Brac{F_{\sigma, \xi} (g(w^*\cdot x) - g(w\cdot x))(x \cdot  (w^* - w) )} \\
			&\qquad + 
		 	\paren{\E_{\xi, \eps }[\sgn(\xi +\eps)] - c}~\E_x \Brac{x \cdot  (w^* - w) }.
		\end{align*}	
		By  \Cref{fact:expectation_eps_xi} and the fact that $g$ is monotone, 
		it follows that $\sgn(F_{\sigma, \xi} (g(w^*\cdot x) - g(w\cdot x))) = \sgn(g(w^*\cdot x) - g(w\cdot x)) = \sgn(x \cdot  (w^* - w) )$ whenever $g(w^* \cdot x) \neq g(w \cdot x)$. 
		Combining this with the fact that $g(\cdot)$ is $1$-Lipschitz, 
            we get 
		\begin{align*}
		&\E_{x} \Brac{F_{\sigma, \xi} (g(w^*\cdot x) - g(w\cdot x))(x \cdot  (w^* - w) )}\\
    &\geq \E_{x} \Brac{F_{\sigma, \xi} (g(w^*\cdot x) - g(w\cdot x))(g(w^*\cdot x) - g(w\cdot x))} \;.
		\end{align*}
		Continuing the calculation above, we see 	
		\begin{align*}
		H_c(w) \cdot (w^* - w) 
		&\geq \E_{x} \Brac{F_{\sigma, \xi} (g(w^*\cdot x) - g(w\cdot x))(g(w^*\cdot x) - g(w\cdot x))} \\
    &\qquad -
		2R~\abs{\E_{\xi, \eps }[\sgn(\xi +\eps)] - c} \;,
		\end{align*}	
		where the bound on the second quantity follows from the fact that $\norm{w}_2, \norm{w^*}_2 \leq R$ and $\norm{x}_2 \leq 1$. 
		\Cref{fact:expectation_eps_xi} implies that $\abs{F_{\sigma, \xi}(t)} \geq \gamma \alpha/4$ 
		if $\abs{t} \geq \gamma \sigma$, whenever $\gamma \leq 2$. 
		We now consider the event $E_{\gamma} := \{x \mid \abs{g(w^* \cdot x) - g(w\cdot x)} \geq \gamma \sigma \}$,
		which describes the region where there is significant difference between the hypothesis $w$ and the target $w^*$. Then we can write 
		\begin{align*}
			H_c(w) \cdot (w^* - w) &\geq  (\gamma \alpha/4) ~\E_{x}\Brac{\abs{(g(w^*\cdot x) - g(w\cdot x))} \mathbf{1}(x \in E_\gamma)} -2R~\abs{\E_{\xi, \eps }[\sgn(\xi +\eps)] - c}\\
			&\geq (\gamma \alpha/4) (\E_x\Brac{\abs{(g(w^*\cdot x) - g(w\cdot x))}} - \gamma \sigma)-
			2R~\abs{\E_{\xi, \eps }[\sgn(\xi +\eps)] - c}\\
			&\geq (\gamma \alpha/4) \E_x\Brac{\abs{(g(w^*\cdot x) - g(w\cdot x))}} -
		(\gamma^2 \alpha \sigma / 4) - 2R~\abs{\E_{\xi, \eps }[\sgn(\xi +\eps)] - c} \;.
		\end{align*}
		In the case that $ \E_x\Brac{\abs{(g(w^*\cdot x) - g(w\cdot x))}} > \Delta$, 
		we would like to set \new{the parameter} $\gamma$ such that $(\gamma^2 \alpha \sigma / 4) + 2R~\abs{\E_{\xi, \eps }[\sgn(\xi +\eps)] - c} \leq \gamma \alpha\Delta/8$, 
		ensuring that the right hand side above is strictly positive. 
		By assumption, we know that $c$ satisfies $\abs{\E_{\xi, \eps }[\sgn(\xi +\eps)] - c} \leq (\gamma \alpha \Delta / 32 R)$, 
		so it suffices for $\gamma$ to satisfy $(\gamma^2 \alpha \sigma / 2) \leq \gamma \alpha\Delta/8$, 
		i.e., $\gamma \leq \Delta / 4\sigma$, in addition to $\gamma \leq 2$. 
		Here, we set $\gamma = \min(\Delta/4\sigma, 1/2)$. 
		Putting these together, we see that when $\Delta \leq 2\sigma$, it holds 	
		\begin{align*}
			H_c(w) \cdot (w^* - w) &\geq (\gamma \alpha/4) \E_x\Brac{\abs{(g(w^*\cdot x) - g(w\cdot x))}} -
			(\gamma^2 \alpha \sigma / 4) - 2R~\abs{\E_{\xi, \eps }[\sgn(\xi +\eps)] - c}\\
			&\geq (\gamma \alpha/4) \E_x\Brac{\abs{(g(w^*\cdot x) - g(w\cdot x))}} -
			(\gamma^2 \alpha \sigma / 4) - (\gamma \alpha \Delta / 16) \\
			&\geq (\alpha \Delta)/(16 \sigma)\E_x\Brac{\abs{(g(w^*\cdot x) - g(w\cdot x))}} - (\alpha \Delta^2)/(64 \sigma) - (\alpha \Delta^2 / 64\sigma).
		\end{align*}
		In the case that we look at a vector $w$ that is $\Delta$-separated from $w^*$, 
		the lower bound we get is $(\alpha \Delta^2)/(32 \sigma)$ when $\Delta \leq 2\sigma$, 
		while the lower bound is $\alpha \Delta / 8$ when $\Delta > 2\sigma$.
	\end{proof}

The following corollary allows us to extend \Cref{lem:GD_lower_bound} to the empirical setting. 
The proof of the corollary can be found in \Cref{app:concentration}. 

\begin{corollary}[Empirical Separating Hyperplane]
	\label{lem:sep_concentration}
	Let $(x_i, y_i)_{i=1}^m \sim \text{GLM-Ob}(g, \sigma, w^*)^m$, 
        where $m \gtrsim R^2 \ln(1/\delta) / (\gamma \alpha \Delta)^2$. 
	Assume $c$ satisfies the assumption in \Cref{lem:GD_lower_bound}. 
	Define $\widehat{H}_c(w) := (1/m)\littlesum_{i=1}^m \Brac{\Paren{\mathrm{sign}(g(w \cdot x_i) - y_i) - c}~x_i}$. 
	Then, for any $w$, it holds 
	\[
	\widehat{H}_c(w) \cdot (w -w^*)  \geq  (\gamma \alpha/4) \E_x\Brac{\abs{(g(w^*\cdot x) - g(w\cdot x))}} 
	- \gamma^2 \alpha \sigma / 4 - 3~(\gamma \alpha \Delta/32)
	\]
	with probability at least $1-\delta$.
\end{corollary}

{While not directly useful in the proof we present here, as pointed out by a reviewer, 
we note that our direction of improvement $\widehat H_c(w)$ as defined in \Cref{lem:sep_concentration} 
can be interpreted to be the gradient of the convex surrogate loss 
$(1/m)~\sum_{i=1}^m\int_0^{w \cdot x_i} (\sgn(g(z) - y_i) + c) ~dz$. This has an analogy to the ``matching loss"  $(1/m)~\sum_{i=1}^m\int_0^{w \cdot x_i} (g(z) - y_i) ~dz$ as considered for the case of $\ell_2$ GLM regression introduced in the work of \citet{Auer:97} and used extensively in subsequent works.}

	\subsection{Reduction to Online Convex Optimization}

	If $c$ is a good approximation of $\E_{\xi, \eps}[\sgn(\xi + \eps)]$, 
	we can reduce the problem to online convex optimization to now get a set of candidates, 
	one of which is close to the true solution. 
	
	\paragraph{OCO Setting} 
	The typical online convex optimization scenario can be modelled as the following game: 
	at time $t-1$ the player must pick a candidate point $w_t$ belonging to a certain constrained set $W$. 
	At time $t$ the true convex loss $f_t(\cdot)$ is revealed and the player suffers a loss of $f_t(w_t)$. 
	This continues for a total of $T$ rounds. 
	Algorithms for these settings typically upper bound the regret ($R(\{w_i\}_{i=1}^T)$), 
	which is the performance with respect to the optimal fixed point in hindsight, 
	$R(\{w_i\}_{i=1}^T) := \sum_{i=1}^T f_t(w_t) - \min_{w^* \in W}\Paren{\sum_{i=1}^T f_t(w^*) }$.  
	
	We specialize Theorem 3.1 from \cite{OCObook} to our setting to get the following lemma.
	
	\begin{lemma}[see, e.g.,~{Theorem 3.1 from \cite{OCObook}}]\label{lem:oco_guarentee}
	Suppose $v_1, \dots, v_T \in \R^d$ such that for all $t \in [T]$ and $\norm{v_t}_2 \leq G$.
	Then online gradient descent with step sizes $\{\eta_t = \frac{R}{G\sqrt{t}} \mid t \in [T]\}$, 
	for linear cost functions $f_t(w) := v_t \cdot w$, outputs a sequence of predictions 
	$w_1, \dots, w_T \in B_d(R)$ such that 
	$\sum_{t=1}^T f_t(w_t) -  \min_{\norm{w}_2 \leq R} \sum_{t=1}^T f_t(w) \leq { (3/2)~GR\sqrt{T}}. $
	\end{lemma}

	An application of this lemma then gives us our result for reducing the problem to OCO.

	\begin{lemma}[Reduction to OCO]\label{thm:OCO_reduction}
		Suppose $(x_1, y_1), \dots, (x_m, y_m)$ are drawn from $\text{GLM-Ob}(g, \sigma, w^*)$ 
		and $c$ satisfies the assumption in \Cref{lem:GD_lower_bound}. 
		Let $T \gtrsim (R / \gamma \alpha)^2$ and 
		$m \gtrsim R^2 \ln(T/\delta) / (\gamma \alpha \Delta)^2$. 
		Then there is an algorithm which recovers a set of candidates $w_1, \dots, w_T$ 
		with probability $1-\delta$ such that 
		\begin{align*}
	 	\min_{w_t} \Set{ \E_x\Brac{\abs{g(w_t \cdot x) - g(w^*\cdot x)}}} \leq 3 \Delta.
		\end{align*}
	\end{lemma}

	\begin{proof}
		At round $t$, the player proposes weight vector $w_t$, 
		at which point the function $f_t(\cdot)$ is revealed to be $f_t(w) := v_t \cdot w$ 
		where $v_t := \widehat{H}_c(w_{t})$ as defined in \Cref{lem:sep_concentration}. 
		Note that a union bound over the $T$ final candidates will ensure that with 
		$m \gtrsim R^2 \ln(T/\delta) / (\gamma \alpha \Delta)^2$ samples, 
		with probability $1-\delta$, 
		for every $t \in [1, T]$, $\widehat{H}_c(w_{t})$ satisfies the conclusion of \Cref{lem:sep_concentration}.
		
		An application of \Cref{lem:oco_guarentee} to this setting gives  
		$$\frac{1}{T}\sum_{t=1}^T f_t(w_t) \leq \min_{\|w\| \le R} \Paren{\frac{1}{T}\sum_{t=1}^T f_t(w)} + 	\frac{(3/2) GR}{\sqrt{T}} \leq \frac{1}{T}\sum_{t=1}^T f_t(w^*) + 	\frac{(3/2) GR}{\sqrt{T}} \;.$$ 
       Rearranging this and applying \Cref{lem:sep_concentration} we get 
		\begin{align*}
			\frac{(3/2) GR}{\sqrt{T}} &\geq \frac{1}{T}\Paren{\sum_{t=1}^T f_t(w_t) - \sum_{t=1}^T f_t(w^*)} = \frac{1}{T} \Paren{\sum_{t=1}^T v_t \cdot \paren{w_t - w^*}}\\
			&\geq \frac{\gamma \alpha}{4T}~\Paren{\sum_{t=1}^T \E_x\Brac{\abs{g(x\cdot w_t) - g(x \cdot w^*)}}} - \gamma^2 \alpha \sigma / 4 - 3~(\gamma \alpha \Delta/32)\\
			&\geq (\gamma \alpha / 4) \min_{w_t} \Set{ \E_x\Brac{\abs{g(x\cdot w_t) - g(x \cdot w^*)}}} - \gamma^2 \alpha \sigma / 4 - 3~(\gamma \alpha \Delta/32) \;, 
		\end{align*}
		where the final inequality follows from the fact that the minimum is smaller than the average.
		Rearranging this gives us $\frac{6 GR}{\gamma \alpha \sqrt{T}} + \gamma \sigma + (3/8) \Delta \geq \min_{w_t} \Set{ \E_x\Brac{\abs{g(x\cdot w_t) - g(x \cdot w^*)}}}$.
		Substituting $\norm{v_t}_2 \leq G = 2$ and $\gamma = \min(\Delta/4\sigma, 1/2)$, 
  we get 
  $$O\Paren{\frac{R}{\gamma \alpha \sqrt{T}}} + 2\Delta \geq \min_{w_t} \Set{ \E_x\Brac{\abs{g(x\cdot w_t) - g(x \cdot w^*)}}} $$ 
  and so, setting $T \gtrsim (R / \gamma \alpha)^2$ ensures that we achieve an error of $3\Delta$.
	\end{proof}
	Note that if the desired lower bound was a convex function (instead of $\E_x [|g(x \cdot w) - g(x\cdot w^*)|]$), 
	we would not have to take the minimum of all the iterates in the proof. 
	We could instead use Jensen's inequality to take the loss of the average iterates. 
	Unfortunately, because the objective can be non-convex due to the nonlinearity of the activation function $g$, 
	we can't just use the averaged iterates.


\section{Pruning Implausible Candidates}

\begin{algorithm}[H]
	\caption{Prune Implausible Candidates}
	\label{alg:approx_search}
	\SetAlgoLined
	\textbf{input:} $\tau, \alpha, \sigma, R, \mathcal{W} = \{w_1, \dots, w_p\}$ \\
	Draw $m = C\log(\abs{\cW}^2/\delta)/(\alpha \tau (\min \{ \tau/2\sigma,1\}))^6$ samples $\{(x_k, y_k)\}_{k=1}^m$ for some constant $C$.\\
	\For{$i \leftarrow 1...p$}{
		\For{$j \leftarrow i+1...p$}{
			
			Let $E_{A}^+ := \{x_k | g(w_i \cdot x_k) - g(w_j \cdot x_k) > A \}$ and 
			$E_{A}^- := \{x_k | g(w_i \cdot x_k) - g(w_j \cdot x_k) < A \}$ \\
			Compute range $U^+$ of $A$ such that $|E_{A+\tau/2}^+| \geq \alpha m \min \{ \tau/2\sigma, 1/4 \}$ via binary search on at most $m$ distinct $g(w_i \cdot x_k) - g(w_j \cdot x_k) - \tau/2$ and similarly $U^-$ for $|E_{A-\tau/2}^-|$\\
			Let $A \leftarrow$ any number in $U^+ \cap U^-$\\
			\If{no such $A$ exists}{
				continue to $(j+1)$-th inner loop
			}
			Compute $R^+ = \{r | r = y_i - g(w_i \cdot x)$ for $x \in E_{A+\tau/2}^{+} \}$ and similarly $R^-$ for $E_{A-\tau/2}^{-}$ \\
			\If{$\abs{\widehat \E[\sgn(r - A) \mid  r \in R^+] -  \widehat \E[\sgn(r - A) \mid  r \in R^-] } >  \alpha \min \{ \tau/16\sigma, 1/8 \}$}{
				reject $w_i$ and continue with $(i + 1)$-th outer loop
			}
		}
		
	}
	\For{$i \leftarrow 1 \dots p$}{
		\For{$j \leftarrow 1 \dots p$}{
			\If{$\frac{1}{m} \sum_{t=1}^m \abs{g(w_i \cdot x_t) - g(w_j \cdot x_t)} > 3 \Delta$}
			{\Return $\cW$}
		}
	}
	Sample $\widehat w$ uniformly from $\cW$.\\
	\Return $\{\widehat w\}$. 
\end{algorithm}

\Cref{thm:OCO_reduction} can generate potential solutions to achieve a low clean loss with respect to $g(w^* \cdot x)$ if $c$ is a good approximation of $\E_{\xi, \eps}[\sgn(\xi + \eps)]$. 
Unfortunately, it is difficult to verify the accuracy of these candidates on the data since it is impossible to differentiate between translations of $g(w^* \cdot x)$ due to the generality of the setting and since $\E_{\xi, \eps}[\sgn(\xi + \eps)]$ is unknown. 
Our algorithm generates $T$ candidates for each value of $c$ in a uniform partition of $[-1, 1]$. 
One the candidates is close to $w^*$, however, the problem of spurious candidates still remains. 
In this section, we discuss how to determine which of the candidate solutions is the best fit for the data.

Even though it is difficult to test if a single hypothesis achieves a small clean loss, it is surprisingly possible to find a good hypothesis out of a list of candidates. 
\Cref{alg:approx_search} describes a tournament-style testing procedure which produces a set of candidates approximately equal to $g(w^* \cdot x)$, and if efficient identifiability holds 
for the instance, 
this list will only contain one candidate. 
The proof of \Cref{lem:prune_bad_candidates} is presented in \Cref{sec:testing}.

\begin{lemma}[Pruning bad candidates]
	Let $\delta > 0$. Suppose $\exists \widehat w \in \mathcal{W}$ such that 
	\[ \E_{x} \Brac{\abs{g(\widehat w \cdot x)- g(w^* \cdot x)}} \leq \min\{\Delta, \tau^2/16\}.\] 
	Then \Cref{alg:approx_search} draws $m \gtrsim \log(\abs{\cW}^2/\delta)/(\alpha^2 \tau^4 (\min \{ \tau/\sigma,1\})^2) + R^2\log(\abs{\cW}^2/\delta) / \Delta^2 $ samples, 
	runs in time $\tilde{O}(dm \abs{\cW}^2)$, 
	and with probability $1-\delta$ returns a list of candidates containing $\widehat w$ such that 
	each candidate satisfies $\Pr[\abs{g(w^* \cdot x) - g(w \cdot x) - A_w} > \tau ] \leq 1-\tau$ 
	for some $A_w \in \R$.
	If $(\Delta, \tau)$-identifiability (\Cref{def:identifiability}) holds, the algorithm only returns a single candidate $\widehat w$ which achieves a clean loss of $4\Delta$.
\end{lemma}

\paragraph{Proof Sketch}
For the sake of exposition, suppose $\widehat w = w^*$ and the empircal estimates equal the true expectation. 
Define the events $E^+_{s}(u, v) := \{ x \mid g(u \cdot x) - g(v \cdot x) > s \}$ and $E^-_{s}(u, v) := \{ x \mid g(u \cdot x) - g(v \cdot x) < s\}$. 
An application of \Cref{fact:small_tail_prob} to the random variable $g(w^* \cdot x) - g(w \cdot x)$ implies that for any $\tau$ if the following first condition is false, then the second condition is true:
\begin{enumerate}
	\item $\exists A \in \R$ such that $\Pr[E^+_{A + \tau}(w^*, w)] \geq \tau / 2$ and $\Pr[E^-_{A - \tau}(w^*, w)] \geq \tau/ 2$. \label{cond1*}
	\item $\exists A \in \R$ such that $\Pr[E^+_{A + \tau}(w^*, w)] \le \tau/ 2$ and $\Pr[E^-_{A - \tau}(w^*, w)] \le \tau/ 2$. \label{cond2*}
\end{enumerate}
If $w$ satisfies Condition~\ref{cond1*}, then  $g(w^* \cdot x) - g(w \cdot x) - A$ takes values $> \tau$ and $\leq \tau$ when $x \in E^+_{A + \tau}(w^*, w)$ and $E^-_{A - \tau}(w^*, w)$ respectively. 
This means the quantile at which $(y - g(w \cdot x))$ takes the value $0$ is different conditioned on $x$ coming from both these sets. 
Let $R^+ := \{ (y - g(w\cdot x)) \mid x \in E^+_{A + \tau}(w^*, w)\}$ and $R^- := \{ (y - g(w\cdot x)) \mid x \in E^-_{A - \tau}(w^*, w)\}$
Our algorithm rejects $w$ if there is an $A$ such that $\abs{ \widehat \E[\sgn(r - A) \mid r \in R^+] - \widehat \E[\sgn(r - A) \mid r \in R^-]}$ is large. This will be the case since elements of $R^+$ and $R^-$ are drawn from the distribution of $\xi + \eps$ shifted by at least $\tau$ in opposite directions, and $\xi$ places a mass of $\alpha$ at $0$. 

Hence, all remaining candidates satisfy Condition~\ref{cond2*}, which means they are approximate translations of $g(w^* \cdot x)$. Also, since $w^*$ is never rejected, we know that $w^*$ also belongs to this list. 
If $(\Delta, \tau)$-identifiability holds, every element of the final list achieves clean loss $\Delta$. We can test this by checking of every pair of candidates in the list is $2\Delta$-close, and if they are, returning any element of the list.


	\section{Main Results}
\begin{algorithm}[H]
	\caption{ Oblivious GLM Regression }
	\label{alg:main_algo}
	\SetAlgoLined
	\textbf{input:} $\{ (x_i, y_i) \mid i \in [m] \} \sim \text{GLM-Ob}(g, \sigma, w^*)^m, R, \sigma, \tau, \alpha$ where $w^*$ is unknown and $\norm{w^*}_2 \leq R$. \\
	Let $P$ be a uniform parititon of $[-1, 1]$ with granularity $\gamma \alpha \Delta /64 R$.\\
	\For{ $c$ in P}
	{ Set the parameter $\Delta$ in \Cref{thm:OCO_reduction} to be $\min(\Delta/3, \tau^2/48)$. \\
	Generate a list of $T$ candidates $\cW_c$ given by each step of the algorithm in \Cref{thm:OCO_reduction}.
	}
	Run \Cref{alg:approx_search} with parameters $\alpha, \sigma, \cup_{c \in P} \cW_c$ to get list $\cL$.\\
	Return $\cL$.
\end{algorithm}

		Finally, we state and prove our two main results. 
		Our first result is a lower bound, demonstrating the necessity of our condition for efficient identifiability. 
		Our second result is our algorithmic guarantee, demonstrating that if efficient identifiability holds, our algorithm returns a hypothesis achieving a small clean loss. 
		\subsection{Necessity of the Identifiability Condition for Unique Recovery}
	\begin{theorem}\label{prop:lower_bound_fin}
		Suppose $\text{GLM-Ob}(g, \sigma, w^*)$ is not $(\Delta ,\tau)$-identifiable, 
		and suppose $u, v \in \R^d$ and $A \in \R$ witness this, 
		i.e. $u, v$ are $\Delta$-separated but satisfy $\Pr_x \Brac{ \abs{ g(u \cdot x) - g(v \cdot x) - A} > \tau } \leq \tau$. 
		Then any algorithm that distinguishes between $u$ and $v$ with probability at least $1-\delta$ 
		requires $m=\Omega(\min(\sigma, 1) \ln(1/\delta) / \tau)$ samples.
	\end{theorem}
	\begin{proof}
		Given $A$ and $\tau$, 
		consider the event $E$ defined by $\abs{ g(u \cdot x_i) - g(v \cdot x_i) - A} > \tau$. 
		This occurs with probability $\leq \tau$. 
		A single sample observed in event $E$ can be enough to tell the difference between $u$ and $v$, 
		and so, to distinguish between $u$ and $v$ with a probability of at least $1-\delta$, 
		one must observe $\Omega(\ln(1/\delta)/\tau)$ samples from $E$. 
		
		If no samples from $E$ are observed, then all $(x_i, y_i)$ satisfy $\abs{ g(u \cdot x_i) - g(v \cdot x_i) - A} \le \tau$. 
		In this case, an oblivious noise adversary can construct oblivious noises $\xi_u$, $\xi_v$ 
		for instances of $u$, $v$ such that the corrupted labels $g(u \cdot x_i) + \xi_u$ and $g(v \cdot x_i) + \xi_v$ 
		only differ by at most $\tau$. 
		This means that $y_i$ can either be generated from $g(u \cdot x_i) + \xi_u + \epsilon$ or $g(v \cdot x_i) + \xi_v + \epsilon$, 
		which are close to each other in total variation distance. 
		By \Cref{lem:TV_lower_bound}, any algorithm to distinguish $u$ and $v$ using inliers requires at least $\Omega(\sigma \ln(1/\delta) / \tau)$ samples. 
		The lower bound corresponds to the minimum of the two sample complexities, so any algorithm to distinguish $u$ and $v$ with probability at least $1-\delta$ needs $\Omega( \min(\sigma, 1) \ln(1/\delta) / \tau)$ samples.
	\end{proof}

	\subsection{Main Algorithmic Result} 
	
	Here, we state the formal version of \Cref{thm:main_informal}. 
	This follows from putting together \Cref{thm:OCO_reduction} and 
	\Cref{lem:prune_bad_candidates}, applied to \Cref{alg:main_algo}. We restate and prove this in \Cref{sec:Proof of Main Theorem}.

\begin{theorem}[Main Result] \label{thm:main-detailed}
We first define a few variables and their relationships to $\Delta$ (the desired final accuracy), $\alpha$ (the probability of being an inlier), $R$ (an upper bound on $\| w^*\|_2$)
and $\sigma$ (the standard deviation of the additive Gaussian noise). 

Let
$\Delta' = \min(\Delta, \tau^2/16)$. 
$\gamma = \min(\Delta/4\sigma, 1/2)$,
$T \gtrsim (R/\gamma \alpha)^2$,
$m_1 \gtrsim R^2 \ln(T/\delta)/(\gamma \alpha \Delta)^2$
and $W \gtrsim T(\gamma \alpha \Delta/64 R)$.

There is an algorithm, which, given $\Delta, \alpha, R$ and $\sigma$
runs in time $O(dTm_1)$, 
draws $m_1 \gtrsim \alpha^{-2} \log(R/\Delta \alpha \delta) \Paren{ R^2\sigma^2 / \Delta^4}$ samples 
from $\text{GLM-Ob}(g, \sigma, w^*)$ and 
returns a $T(\gamma \alpha \Delta/64 R)$-sized list of candidates, 
one of which achieves excess loss at most $\Delta$. 

Moreover, if the instance is $(\Delta, \tau)$-identifiable
then, there is an algorithm which takes the parameters 
$\Delta, \alpha, \sigma, R$ and $\tau' \leq \tau$,
draws $$m \gtrsim \alpha^{-2} \log(W/\delta) \Paren{ R^2 \sigma^2/ ( \Delta'^4 + 1/(\tau' \min (\tau'/\sigma,1 ))^{2} }$$ samples from $\text{GLM-Ob}(g, \sigma, w^*)$, runs in time  $O(d m W^2)$ and returns a single candidate. 
\end{theorem}
\begin{remark}
\new{

Our pruning algorithm requires a priori knowledge of the parameter $\tau^*$.
However, it is possible to make it independent of $\tau^*$ by guessing the value of some $\tau < \tau^*$ via the following process. 
    
To find such a $\tau$, we can search for a value of $k$ that allows us to set $\tau = 2^{-k}$ in the pruning algorithm and obtain a single potential candidate.
To begin, we decide on the number of tests we are willing to perform, which we will denote by $K$. Each test involves performing pruning and checking if the resulting list contains only one element.

Each test provides the expected answer for that particular $\tau$ with probability 
of $1-\delta$. In other words, with probability $1-\delta$, $\widehat{w}$ is never rejected, and if $\tau < \tau^*$ the test will yield a single such candidate.
By increasing the number of samples by a factor of $O(\log(K))$, we can guarantee that, with probability of $1-\delta$, we will either obtain a single favorable candidate or a list of candidates among which at least one achieves an excess loss of $\Delta$.
}
\end{remark} 
	
	\bibliographystyle{alpha}
	\bibliography{allrefs}

\newcommand{\etalchar}[1]{$^{#1}$}
\begin{thebibliography}{DKMR22b}

\bibitem[AHW95]{Auer95}
P.~Auer, M.~Herbster, and M.~K.~K Warmuth.
\newblock Exponentially many local minima for single neurons.
\newblock In D.~Touretzky, M.C. Mozer, and M.~Hasselmo, editors, {\em Advances
  in Neural Information Processing Systems}, volume~8. MIT Press, 1995.

\bibitem[Aue97]{Auer:97}
P.~Auer.
\newblock Learning nested differences in the presence of malicious noise.
\newblock {\em Theoretical Computer Science}, 185(1):159--175, 1997.

\bibitem[BJK15]{BhatiaJK15}
K.~Bhatia, P.~Jain, and P.~Kar.
\newblock Robust regression via hard thresholding.
\newblock In {\em Advances in Neural Information Processing Systems 28: Annual
  Conference on Neural Information Processing Systems 2015}, pages 721--729,
  2015.

\bibitem[BJKK17]{BhatiaJKK17}
K.~Bhatia, P.~Jain, P.~Kamalaruban, and P.~Kar.
\newblock Consistent robust regression.
\newblock In {\em Advances in Neural Information Processing Systems 30: Annual
  Conference on Neural Information Processing Systems 2017}, pages 2107--2116,
  2017.

\bibitem[Cd22]{HT22}
H.~Chen and T.~d'Orsi.
\newblock On the well-spread property and its relation to linear regression.
\newblock In Po-Ling Loh and Maxim Raginsky, editors, {\em Proceedings of
  Thirty Fifth Conference on Learning Theory}, volume 178 of {\em Proceedings
  of Machine Learning Research}, pages 3905--3935. PMLR, 02--05 Jul 2022.

\bibitem[CKMY20a]{chen2020classification}
S.~Chen, F.~Koehler, A.~Moitra, and M.~Yau.
\newblock Classification under misspecification: Halfspaces, generalized linear
  models, and evolvability.
\newblock {\em Advances in Neural Information Processing Systems},
  33:8391--8403, 2020.

\bibitem[CKMY20b]{chen2020online}
S.~Chen, F.~Koehler, A.~Moitra, and M.~Yau.
\newblock Online and distribution-free robustness: Regression and contextual
  bandits with huber contamination.
\newblock {\em arXiv preprint arXiv:2010.04157}, 2020.

\bibitem[DGK{\etalchar{+}}20]{DGKKS20}
I.~Diakonikolas, S.~Goel, S.~Karmalkar, A.~R. Klivans, and M.~Soltanolkotabi.
\newblock Approximation schemes for relu regression.
\newblock In {\em Conference on Learning Theory, {COLT} 2020}, volume 125 of
  {\em Proceedings of Machine Learning Research}, pages 1452--1485. {PMLR},
  2020.

\bibitem[DK22]{DK20-SQ-Massart}
I.~Diakonikolas and D.~M. Kane.
\newblock Near-optimal statistical query hardness of learning halfspaces with
  massart noise.
\newblock In {\em Conference on Learning Theory}, volume 178 of {\em
  Proceedings of Machine Learning Research}, pages 4258--4282. {PMLR}, 2022.

\bibitem[DK23]{DiaKan22-book}
I.~Diakonikolas and D.~M. Kane.
\newblock {\em Algorithmic High-Dimensional Robust Statistics}.
\newblock Cambridge University Press, 2023.

\bibitem[DKK{\etalchar{+}}16]{DKKLMS16}
I.~Diakonikolas, G.~Kamath, D.~M. Kane, J.~Li, A.~Moitra, and A.~Stewart.
\newblock Robust estimators in high dimensions without the computational
  intractability.
\newblock In {\em Proc.\ 57th IEEE Symposium on Foundations of Computer Science
  (FOCS)}, pages 655--664, 2016.

\bibitem[DKMR22a]{DKMR22-massart}
I.~Diakonikolas, D.~M. Kane, P.~Manurangsi, and L.~Ren.
\newblock Cryptographic hardness of learning halfspaces with massart noise.
\newblock {\em CoRR}, abs/2207.14266, 2022.
\newblock Conference version in NeurIPS'22.

\bibitem[DKMR22b]{DKMR22}
I.~Diakonikolas, D.~M. Kane, P.~Manurangsi, and L.~Ren.
\newblock Hardness of learning a single neuron with adversarial label noise.
\newblock In {\em International Conference on Artificial Intelligence and
  Statistics, {AISTATS} 2022}, volume 151 of {\em Proceedings of Machine
  Learning Research}, pages 8199--8213. {PMLR}, 2022.

\bibitem[DKN20]{DKZ20-sq-reg}
I.~Diakonikolas, D.~M. Kane, and N.Zarifis.
\newblock Near-optimal {SQ} lower bounds for agnostically learning halfspaces
  and relus under gaussian marginals.
\newblock In {\em Advances in Neural Information Processing Systems 33: Annual
  Conference on Neural Information Processing Systems 2020, NeurIPS 2020},
  2020.

\bibitem[DKPZ21]{DKP21-SQ}
I.~Diakonikolas, D.~M. Kane, T.~Pittas, and N.~Zarifis.
\newblock The optimality of polynomial regression for agnostic learning under
  gaussian marginals in the sq model.
\newblock {\em Proceedings of Machine Learning Research vol}, 134:1--33, 2021.

\bibitem[DKR23]{DKR23}
I.~Diakonikolas, D.~M. Kane, and L.~Ren.
\newblock Near-optimal cryptographic hardness of agnostically learning
  halfspaces and relu regression under gaussian marginals.
\newblock {\em CoRR}, abs/2302.06512, 2023.
\newblock Conference version in ICML'23.

\bibitem[DKRS22]{DKRS22}
I.~Diakonikolas, D.~M. Kane, L.~Ren, and Y.~Sun.
\newblock {SQ} lower bounds for learning single neurons with massart noise.
\newblock In {\em NeurIPS}, 2022.

\bibitem[DKTZ22]{DKTZ22}
I.~Diakonikolas, V.~Kontonis, C.~Tzamos, and N.~Zarifis.
\newblock Learning a single neuron with adversarial label noise via gradient
  descent.
\newblock In {\em Conference on Learning Theory}, volume 178 of {\em
  Proceedings of Machine Learning Research}, pages 4313--4361. {PMLR}, 2022.

\bibitem[dLN{\etalchar{+}}21]{d2021consistent}
T.~d'Orsi, C.~H. Liu, R.~Nasser, G.~Novikov, D.~Steurer, and S.~Tiegel.
\newblock Consistent estimation for pca and sparse regression with oblivious
  outliers.
\newblock {\em Advances in Neural Information Processing Systems},
  34:25427--25438, 2021.

\bibitem[dNNS22]{SoSOblivious21}
T.~d'Orsi, R.~Nasser, G.~Novikov, and D.~Steurer.
\newblock Higher degree sum-of-squares relaxations robust against oblivious
  outliers.
\newblock {\em CoRR}, abs/2211.07327, 2022.

\bibitem[dNS21]{Steurer21Outliers}
T.~d'Orsi, G.~Novikov, and D.~Steurer.
\newblock Consistent regression when oblivious outliers overwhelm.
\newblock In Marina Meila and Tong Zhang, editors, {\em Proceedings of the 38th
  International Conference on Machine Learning, {ICML} 2021, 18-24 July 2021,
  Virtual Event}, volume 139 of {\em Proceedings of Machine Learning Research},
  pages 2297--2306. {PMLR}, 2021.

\bibitem[DPT21]{diakonikolas2021relu}
I.~Diakonikolas, J.~H. Park, and C.~Tzamos.
\newblock Relu regression with massart noise.
\newblock {\em Advances in Neural Information Processing Systems},
  34:25891--25903, 2021.

\bibitem[DT19]{dalalyan2019outlier}
A.~Dalalyan and P.~Thompson.
\newblock Outlier-robust estimation of a sparse linear model using
  $\ell_1$-penalized huber's $ m $-estimator.
\newblock {\em Advances in neural information processing systems}, 32, 2019.

\bibitem[GGK20]{GGK20}
S.~Goel, A.~Gollakota, and A.~R. Klivans.
\newblock Statistical-query lower bounds via functional gradients.
\newblock In {\em Advances in Neural Information Processing Systems 33: Annual
  Conference on Neural Information Processing Systems 2020, NeurIPS 2020},
  2020.

\bibitem[GKK19]{GoelKK19}
S.~Goel, S.~Karmalkar, and A.~R. Klivans.
\newblock Time/accuracy tradeoffs for learning a relu with respect to gaussian
  marginals.
\newblock In {\em Advances in Neural Information Processing Systems 32: Annual
  Conference on Neural Information Processing Systems 2019, NeurIPS 2019},
  pages 8582--8591, 2019.

\bibitem[Haz16]{OCObook}
E.~Hazan.
\newblock Introduction to online convex optimization.
\newblock {\em Found. Trends Optim.}, 2(3–4):157–325, aug 2016.

\bibitem[HM13]{HardtM13}
M.~Hardt and A.~Moitra.
\newblock Algorithms and hardness for robust subspace recovery.
\newblock In {\em Proc.\ 26th Annual Conference on Learning Theory (COLT)},
  pages 354--375, 2013.

\bibitem[KKP17]{kkp17}
D.~M. Kane, S.~Karmalkar, and E.~Price.
\newblock Robust polynomial regression up to the information theoretic limit.
\newblock In {\em 2017 IEEE 58th Annual Symposium on Foundations of Computer
  Science (FOCS)}, pages 391--402, 2017.

\bibitem[KKSK11]{kakade2011efficient}
S.~M. Kakade, V.~Kanade, O.~Shamir, and A.~Kalai.
\newblock Efficient learning of generalized linear and single index models with
  isotonic regression.
\newblock {\em Advances in Neural Information Processing Systems}, 24, 2011.

\bibitem[KM17]{klivans2017learning}
A.~Klivans and R.~Meka.
\newblock Learning graphical models using multiplicative weights.
\newblock In {\em 2017 IEEE 58th Annual Symposium on Foundations of Computer
  Science (FOCS)}, pages 343--354. IEEE, 2017.

\bibitem[KMM20]{karmakar2020study}
S.~Karmakar, A.~Mukherjee, and R.~Muthukumar.
\newblock A study of neural training with iterative non-gradient methods.
\newblock {\em arXiv e-prints}, pages arXiv--2005, 2020.

\bibitem[KP19]{KarmalkarP19}
S.~Karmalkar and E.~Price.
\newblock Compressed sensing with adversarial sparse noise via {L1} regression.
\newblock In Jeremy~T. Fineman and Michael Mitzenmacher, editors, {\em 2nd
  Symposium on Simplicity in Algorithms, {SOSA} 2019}. Schloss Dagstuhl -
  Leibniz-Zentrum f{\"{u}}r Informatik, 2019.

\bibitem[KS09]{kalai2009isotron}
A.~T. Kalai and R.~Sastry.
\newblock The isotron algorithm: High-dimensional isotonic regression.
\newblock In {\em Conference on Learning Theory (COLT)}, 2009.

\bibitem[KSA19]{kalan2019fitting}
S.~M.~M. Kalan, M.~Soltanolkotabi, and S.~Avestimehr.
\newblock Fitting relus via sgd and quantized sgd.
\newblock In {\em 2019 IEEE International Symposium on Information Theory
  (ISIT)}, pages 2469--2473. IEEE, 2019.

\bibitem[LDB09]{laska2009exact}
J.~N. Laska, M.~A. Davenport, and R.~G. Baraniuk.
\newblock Exact signal recovery from sparsely corrupted measurements through
  the pursuit of justice.
\newblock In {\em 2009 Conference Record of the Forty-Third Asilomar Conference
  on Signals, Systems and Computers}, pages 1556--1560. IEEE, 2009.

\bibitem[Li11]{Li2011CompressedSA}
X.~Li.
\newblock Compressed sensing and matrix completion with constant proportion of
  corruptions.
\newblock {\em Constructive Approximation}, 37:73--99, 2011.

\bibitem[LRV16]{LaiRV16}
K.~A. Lai, A.~B. Rao, and S.~Vempala.
\newblock Agnostic estimation of mean and covariance.
\newblock In {\em Proc.\ 57th IEEE Symposium on Foundations of Computer Science
  (FOCS)}, pages 665--674, 2016.

\bibitem[MR18]{MR18}
P.~Manurangsi and D.~Reichman.
\newblock The computational complexity of training relu (s).
\newblock {\em arXiv preprint arXiv:1810.04207}, 2018.

\bibitem[NT22]{NasserT22}
R.~Nasser and S.~Tiegel.
\newblock Optimal {SQ} lower bounds for learning halfspaces with massart noise.
\newblock In {\em Conference on Learning Theory}, volume 178 of {\em
  Proceedings of Machine Learning Research}, pages 1047--1074. {PMLR}, 2022.

\bibitem[NTN11]{NTN11}
N.~Nasrabadi, T.~Tran, and N.~Nguyen.
\newblock Robust lasso with missing and grossly corrupted observations.
\newblock {\em Advances in Neural Information Processing Systems}, 24, 2011.

\bibitem[NW72]{nelder1972generalized}
J.~A. Nelder and R.~W.~M. Wedderburn.
\newblock Generalized linear models.
\newblock {\em Journal of the Royal Statistical Society: Series A (General)},
  135(3):370--384, 1972.

\bibitem[NWL22]{norman2022robust}
T.~Norman, N.~Weinberger, and K.~Y. Levy.
\newblock Robust linear regression for general feature distribution.
\newblock {\em arXiv preprint arXiv:2202.02080}, 2022.

\bibitem[PF20]{pesme2020online}
S.~Pesme and N.~Flammarion.
\newblock Online robust regression via sgd on the l1 loss.
\newblock {\em Advances in Neural Information Processing Systems},
  33:2540--2552, 2020.

\bibitem[SBRJ19]{suggala2019adaptive}
A.~S. Suggala, K.~Bhatia, P.~Ravikumar, and P.~Jain.
\newblock Adaptive hard thresholding for near-optimal consistent robust
  regression.
\newblock In {\em Conference on Learning Theory}, pages 2892--2897. PMLR, 2019.

\bibitem[Sol17]{Mahdi17}
M.~Soltanolkotabi.
\newblock Learning relus via gradient descent.
\newblock In {\em Advances in Neural Information Processing Systems}, pages
  2007--2017, 2017.

\bibitem[TJSO14]{Tsakonas14}
E.~Tsakonas, J.~Jaldén, N.~D. Sidiropoulos, and B.~Ottersten.
\newblock Convergence of the huber regression m-estimate in the presence of
  dense outliers.
\newblock {\em IEEE Signal Processing Letters}, 21(10):1211--1214, 2014.

\bibitem[WM10]{wright2008dense}
J.~Wright and Y.~Ma.
\newblock Dense error correction via $\ell_1$-minimization.
\newblock {\em IEEE Transactions on Information Theory}, 56(7):3540--3560,
  2010.

\bibitem[WZDD23]{WZDD23}
P.~Wang, N.~Zarifis, I.~Diakonikolas, and J.~Diakonikolas.
\newblock Robustly learning a single neuron via sharpness.
\newblock {\em CoRR}, abs/2306.07892, 2023.

\bibitem[YO20]{yehudai2020learning}
G.~Yehudai and S.~Ohad.
\newblock Learning a single neuron with gradient methods.
\newblock In {\em Conference on Learning Theory}, pages 3756--3786. PMLR, 2020.

\end{thebibliography}
	
	\newpage
	\appendix
\appendix

\section{Concentration and Anti-Concentration}\label{app:concentration}

\begin{lemma}[Hoeffding]\label{lem:hoeffding} 
	Let $X_1, \dots X_n$ be independent random variables such that $X_i \in [a_i, b_i]$. Then $S_n := \frac{1}{n} \sum_{i=1}^{n} X_i$, then for all $t > 0$
	\begin{align*}
		\Pr[\abs{S_n - \E[S_n]} \geq t] \leq \exp\Paren{-\frac{2n^2 t^2}{\sum_{i=1}^n (b_i - a_i)^2}}
	\end{align*}
\end{lemma}


\begin{lemma}[Empirical Separating Hyperplane]
	Let $(x_i, y_i)_{i=1}^m \sim \text{GLM-Ob}(g, \sigma, w^*)^m$ where $m \gtrsim  R^2 \ln(1/\delta) / (\gamma \alpha \Delta)^2$. Assume $c$ satisfies the assumption in \Cref{lem:GD_lower_bound}. Define $\widehat{H}_c(w) := (1/m)\littlesum_{i=1}^m \Brac{\Paren{\mathrm{sign}(g(w \cdot x_i) - y_i) - c}~x_i}$. Then for any $w$,
	\[
	\widehat{H}_c(w) \cdot (w -w^*)  \geq  (\gamma \alpha/4) \E_x\Brac{\abs{(g(w^*\cdot x) - g(w\cdot x))}} 
	- \gamma^2 \alpha \sigma / 4 - 3~(\gamma \alpha \Delta/32)
	\]
	with probability at least $1-\delta$.
\end{lemma}
\begin{proof}
	From Lemma~\ref{lem:GD_lower_bound}, we know that $H_c(w) \cdot (w - w^*) \geq (\gamma \alpha/4) \E_x\Brac{\abs{(g(w^*\cdot x) - g(w\cdot x))}} - \gamma^2 \alpha \sigma / 2$ where $H_c(w) := \E_{x, y}\Brac{\Paren{ \mathrm{sign}(y - g(w \cdot x)) - c}~x}$. 
	Consider the random variable given by 
	$H_c(w)\cdot v - \widehat{H}_c(w) \cdot v$ for any fixed vector $v$. Upon examination, we can see that the quantity $\Paren{\mathrm{sign}(g(x\cdot x_i) - y_i) - c}x_i\cdot v$
	has bounded absolute value at most $2\norm{v} \leq 4R$ because $|c| \leq 1$. Then the concentration follows from a simple application of Hoeffding's inequality (\Cref{lem:hoeffding}).
	\begin{align*}
		\Pr\big[ H_c(w)\cdot v - \widehat{H}_c(w) \cdot v \geq t \big] \leq \exp \left(-\frac{mt^2}{8 R^2} \right) 
	\end{align*}
	Then, setting $v = w-w^*$, $t = \gamma \alpha \Delta/ 32$ and $m = C\Paren{ R^2 \ln(1/\delta) / (\gamma \alpha \Delta)^2}$ for some large enough constant $C$, we have that
	\[
	\widehat{H}_c(w) \cdot (w -w^*)  \geq  (\gamma \alpha/4) \E_x\Brac{\abs{(g(w^*\cdot x) - g(w\cdot x))}} 
	- \gamma^2 \alpha \sigma / 4 - 3~(\gamma \alpha \Delta/32) \;.
	\]	
	\noindent

\end{proof}

\section{Proofs of Basic Facts}\label{app:facts}

\begin{fact}\label{fact:approx_q}
	Given estimates $\widehat a, \widehat b$ of quantities $a, b$ satisfying $0 \le a \le 1$, $L \le b \le 1$ and $\abs{\widehat a - a} \le e$ and $\abs{\widehat b - b} \le e$ where $e \le L/2$, the quotient $\widehat a /\widehat b$ satisfies $\abs{(\widehat a/\widehat b) - (a/b)} \leq 8e/L^2$. 
\end{fact}
\begin{proof}
	We see that $(a/b) - (\widehat a/\widehat b)  \leq (a/b) - (a-e/(b+e)) \leq (ea + eb)/(b(b-e)) \leq 8e/L^2$. The other direction follows by a similar argument. 
\end{proof}

\begin{fact}\label{fact:h_sig}
	Let $\eps \sim \cN(0, \sigma^2)$, then $h_\sigma(t) : \R \rightarrow \R$ defined as $h_\sigma(t) := \E_\eps[\sgn(t + \eps)]$ satisfies:
\begin{enumerate}
	\item $h_\sigma(-t) = -h_\sigma(t)$. 
	\item $h_\sigma(t)$ is strictly increasing.
	\item $\abs{h_\sigma(t)} \leq 1$. 
	\item For every $\tau < 2$, For all $t \notin [-\tau \sigma, \tau \sigma]$, $\abs{h_\sigma(t)} \geq (\tau/4)$, and whenever $\abs{t} \leq \tau \sigma$, $\abs{h_\sigma(t)} \geq (1/4) (t/\sigma)$. 
\end{enumerate}
\end{fact}
\begin{proof}
Suppose $\sigma \neq 0$, if $\sigma = 0$ these properties follow from properties of the sign function. 

The first three follow easily from the fact that 
\[h_\sigma(t) = \Pr_\eps[t + \eps > 0] - \Pr_\eps[t + \eps \leq 0] = \sgn(t)~\Pr_\eps[-|t| \leq \eps \leq |t|] = \sgn(t)~(1-2\Pr_\eps [\eps > |t|]).\]

To see the final property, observe that
\begin{align*}
	h_\sigma(t) &=  \sgn(t)~\Pr_\eps[-|t| \leq \eps \leq |t|]=  \sgn(t)~\Pr_{x \sim \cN(0, 1)}[- t/\sigma \leq x \leq t/\sigma].
\end{align*}
By Gaussian anticoncentration, whenever $t/\sigma < 2$, $\Pr_{x \sim \cN(0, 1)}[- t/\sigma \leq x \leq t/\sigma] \geq (t/4\sigma)$ proving the second part of this claim. Since $h$ is strictly increasing, we see that whenever $\abs{t} > \tau \sigma$ and $\tau < 2$ $\abs{h_\sigma(t)} \geq \tau/4$, proving the first part of the claim. 
\end{proof}

\begin{fact}
Let $\xi$ be oblivious noise such that $\Pr[\xi = 0] \ge \alpha$, then $F_{\sigma, \xi}(t) := \E_{\eps, \xi}[\sgn(t + \eps + \xi)] -  \E_{\eps, \xi}[\sgn(\eps + \xi)]$ satisfies the following: 
\begin{enumerate}
	\item $F_{\sigma, \xi}$ is strictly increasing.
	\item $\sgn(F_{\sigma, \xi}(t)) = \sgn(t)$.
	\item For any $\tau \leq 2$, Whenever $|t| \geq \sigma \tau$, $\abs{F_{\sigma, \xi}(t) } > (\tau \alpha/4)$ and whenever $\abs{t} \leq \sigma \tau$, $\abs{F_{\sigma, \xi}(t) } \leq (\alpha/4) (t/\sigma)$
\end{enumerate}
\end{fact}
\begin{proof}
The first property follows from the fact that if $t_1  - t_2 > 0$ then $F_{\sigma, \xi}(t_1) - F_{\sigma, \xi}(t_2) = \E_\xi \Brac{h_\sigma(t_1 + \xi) - h_\sigma(t_2 + \xi)} > 0$. Hence $F_{\sigma, \xi}(t)$ is strictly increasing.

The second property follows by definition and the first property, $F_{\sigma, \xi} (0) = 0$ and since $F_{\sigma, \xi}$ is strictly increasing, $\sgn(F_{\sigma, \xi}(t)) = \sgn(t)$. 

Note that $h_\sigma$ is strictly increasing. Let $a > c\sigma$.
\begin{align*}
	F_{\sigma, \xi}(a) 
	&=  \E_\xi \Brac{h_\sigma(a + \xi) - h_\sigma(\xi)}\\
	& > \alpha \E_{\xi \mid \xi = 0} \Brac{h_\sigma(a) - h_\sigma(0)}\\
	& = \alpha h_\sigma(a) \;.
\end{align*}
The first inequality above follows from the fact that $h_\sigma$ is montone and $a > 0$. The final property above now follows from Property 4 of \Cref{fact:h_sig}. A similar argument holds when $a < -c\sigma$. 
\end{proof}

\begin{remark}
Note that a similar result holds for other distributions as long as the measurement noise
	has some density around the origin. 
	More precisely, if $\Pr_{\eps}[\abs{\eps} \leq \sigma] \geq C$ 
	then $\Pr_{\eps, \xi} [\abs{\eps + \xi} \leq \sigma] \geq \alpha C$. 
	This implies that $F_{\sigma, \xi}(t)$ as defined above
	satisfies $\abs{F_{\sigma, \xi}(t)} \geq C\alpha$ whenever $\abs{t}\geq \sigma$,
	for $\eps$ satisfying the constraint  $\Pr_{\eps}[\abs{\eps} \leq \sigma] \geq C$. 
	Hence, the Gaussianity of our observation noise is not crucial, but
	the noise needs to have some density around the origin.	
	Indeed, the oblivious noise is free to incorporate any other distribution as well.
\end{remark}

\begin{fact}\label{lem:TV_lower_bound}
	Let $\mathbf p$ and $\mathbf q$ be univariate probability distributions on $\R$ and denote total variation distance as $d_{TV}$. Then any algorithm requires $\Omega(\ln(1/\delta)/d_{TV}(\mathbf p, \mathbf q))$ samples to successfully distinguish between $\mathbf p, \mathbf q$ with probability $1-\delta$.
\end{fact}

\begin{fact}
	Let $X$ be a random variable on $\R$. Fix $\tau > 0$ and $\eta > 0$. Define the events $E_A^+$ and $E_A^-$ such that
	$\Pr[E_A^+] = \Pr[X > A + \tau]$ and $\Pr[E_A^-] = \Pr[X < A - \tau]$. Then if the first condition below is not true, the second is.
	\begin{enumerate}
		\item $\exists A \in \R$ such that $\Pr[E_A^+] \geq \eta$ and $\Pr[E_A^-] \geq \eta$.
		\item $\exists A^* \in \R$ such that $\Pr[E_{A^*}^+] \le \eta$ and $\Pr[E_{A^*}^-] \le \eta$.
	\end{enumerate}
\end{fact}
\begin{proof}
	Assume the first statement is false. Then the negation implies that $\forall A \in \R$, either
	$\Pr[E_A^+] \leq \eta$ or $\Pr[E_A^-] \leq \eta$. We want to show that the ``or'' statement translates to 
	an ``and'' statement for a particular $A^*$.

	Note that $\Pr[E_A^-]$ as a function of $A$ can be seen as the CDF of $X$ without the equality portion where $X = A-\tau$.
	This means that $\Pr[E_A^-]$ is left-continuous with respect to $A$. Therefore, we can define $A^*$ such that
	$\forall A \in (-\infty, A^*]$, $\Pr[E_A^-] \leq \eta$ and for any $A > A^*$, $\Pr[E_A^-] > \eta$.
	Then, by our initial assumption, it must be the case that $\forall A \in (A^*, \infty)$, $\Pr[E_A^+] \leq \eta$.
	In contrast to $\Pr[E_A^-]$ being left-continuous, we can infer that $\Pr[E_A^+]$ as a function of $A$ is
	right-continuous with respect to $A$. Therefore by right-continuity, $\Pr[E_{A^*}^+] \leq \eta$. This proves the existence
	of such $A^*$ of the second condition and concludes the proof.
	
\end{proof}

\section{Pruning Implausible Solutions}\label{sec:testing}

\new{Here we provide a detailed analysis of our pruning procedure.}

\begin{algorithm}[H]
	\caption{Prune Implausible Candidates}
	\label{alg:approx_search}
	\SetAlgoLined
	\textbf{input:} $\alpha, \sigma, R, \mathcal{W} = \{w_1, \dots, w_p\}, \tau$ \\
	Draw $m = C\log(\abs{\cW}^2/\delta)/(\alpha \tau (\min \{ \tau/2\sigma,1\}))^2$ samples $\{(x_k, y_k)\}_{k=1}^m$ for some constant $C$.\\
	\For{$i \leftarrow 1...p$}{
		\For{$j \leftarrow i+1...p$}{
			
			Let $E_{A}^+ := \{x_k | g(w_i \cdot x_k) - g(w_j \cdot x_k) > A \}$ and 
			$E_{A}^- := \{x_k | g(w_i \cdot x_k) - g(w_j \cdot x_k) < A \}$ \\
			Compute the range of $A$ such that $|E_{A+\tau/2}^+| \geq \alpha m \min \{ \tau/2\sigma, 1/4 \}$ via binary search on at most $m$ distinct $g(w_i \cdot x_k) - g(w_j \cdot x_k) - \tau/2$ and similarly for $|E_{A-\tau/2}^-|$\\
			Let $A \leftarrow$ any number in the intersection of two ranges\\
			\If{no such $A$ exists}{
				continue to $(j+1)$-th inner loop
			}
			Compute $R^+ = \{r | r = y_i - g(w_i \cdot x)$ for $x \in E_{A+\tau/2}^{+} \}$ and similarly $R^-$ for $E_{A-\tau/2}^{-}$ \\
			\If{$\abs{\widehat \E[\sgn(r - A) \mid  r \in R^+] -  \widehat \E[\sgn(r - A) \mid  r \in R^-] } >  \alpha \min \{ \tau/16\sigma, 1/8 \}$}{
				reject $w_i$ and continue with $(i + 1)$-th outer loop
			}
		}
		
	}
	\For{$i \leftarrow 1 \dots p$}{
		\For{$j \leftarrow 1 \dots p$}{
			\If{$\frac{1}{m} \sum_{t=1}^m \abs{g(w_i \cdot x_t) - g(w_j \cdot x_t)} > 3 \Delta$}
			{\Return $\cW$}
		}
	}
	Sample $\widehat w$ uniformly from $\cW$.\\
	\Return $\{\widehat w\}$. 
\end{algorithm}

\begin{lemma}[Pruning bad candidates]\label{lem:prune_bad_candidates}
	Let $\delta > 0$. Suppose $\exists \widehat w \in \mathcal{W}$ such that 
	\[ \E_{x} \Brac{\abs{g(\widehat w \cdot x)- g(w^* \cdot x)}} \leq \min\{\Delta, \tau^2/16\}.\] 
	Then \Cref{alg:approx_search} draws $m \gtrsim \log(\abs{\cW}^2/\delta)/(\alpha^2 \tau^4 (\min \{ \tau/\sigma,1\})^2) + R^2\log(\abs{\cW}^2/\delta) / \Delta^2 $ samples, 
	runs in time $\tilde{O}(d m \abs{\cW}^2)$, 
	and with probability $1-\delta$ returns a list of candidates containing $\widehat w$ such that 
	each candidate satisfies $\Pr[\abs{g(w^* \cdot x) - g(w \cdot x) - A_w} > \tau ] \leq 1-\tau$ 
	for some $A_w \in \R$.
	If $(\Delta, \tau)$-identifiability (\Cref{def:identifiability}) holds, the algorithm only returns a single candidate $\widehat w$ which achieves a clean loss of $4\Delta$.
\end{lemma}

\begin{proof} 
	Define the events $E^+_{s}(\widehat w, w) := \{ x \mid g(\widehat w \cdot x) - g(w \cdot x) > s \}$ 
	and $E^-_{s}(\widehat w, w) := \{ x \mid g(\widehat w \cdot x) - g(w \cdot x) < s\}$. 
	An application of \Cref{fact:small_tail_prob} to the random variable $g(\widehat w \cdot x) - g(w \cdot x)$ implies that for any $\tau_0, \eta_0$ if the first condition below is not true, the second is.
	\begin{enumerate}
		\item $\exists A \in \R$ such that $\Pr[E^+_{A + \tau_0}(\widehat w, w)] \geq \eta_0$ and $\Pr[E^-_{A - \tau_0}(\widehat w, w)] \geq \eta_0$. \label{cond1}
		\item $\exists A \in \R$ such that $\Pr[E^+_{A + \tau_0}(\widehat w, w)] \le \eta_0$ and $\Pr[E^-_{A - \tau_0}(\widehat w, w)] \le \eta_0$. \label{cond2}
	\end{enumerate}
	
	In the first part of our proof, we show that the algorithm rejects $w$ if Condition~\ref{cond1} holds.
	We will need the following lemma about $R^+$ and $R^-$ as defined in our algorithm. 
	
	\begin{claim}\label{claim:r_ppties}
		Suppose $C := \{x \mid \abs{g( \widehat w \cdot x)- g(w^* \cdot x)} \leq \tau_0/2 \}$. 
		Then for any choice of $\tau_0$ and $\eta_0$,  
		if Condition~\ref{cond1} holds, 
		there is a choice of $\delta' = 2\Delta/\tau_0$ satisfying,  
		\begin{enumerate}
			\item $\Pr[C \mid E^+_{A + \tau_0}(\widehat w, w) ] \geq 1- \delta'/\eta_0$. 
			\item $\max \{ \Pr[\overline C \mid E^+_{A + \tau_0}(\widehat w, w)], \Pr[ \overline C \mid E^-_{A - \tau_0}(\widehat w, w)] \} \leq \delta'/\eta_0$.
			\item $x \in E^+_{A + \tau_0}(\widehat w, w) \cap C$ implies $\sgn(y(x) - g(w \cdot x) - A) \geq \sgn(\eps + \xi + \tau_0/2)$ and $x \in E^-_{A - \tau_0}(\widehat w, w) \cap C$ implies $\sgn(y(x) - g(w \cdot x) - A) \leq \sgn(\eps + \xi - \tau_0/2)$.
		\end{enumerate}
	\end{claim}
	\begin{proof}
		Since Condition~\ref{cond1} holds, $\Pr[E^+_{A + \tau_0}(\widehat w, w)] \geq \eta_0$ and $\Pr[E^-_{A - \tau_0}(\widehat w, w)] \geq \eta_0$. 
		
		We now lower bound the probability of $C$. 
		By assumption, $\E_{x} \Brac{\abs{g( \widehat w \cdot x)- g(w^* \cdot x)}}  \leq \Delta$. 
		An application of Markov's inequality implies $\Pr\brac{\abs{g(\widehat w \cdot x)- g(w^* \cdot x)} \geq \tau_0/2} \leq 2\Delta/\tau_0$. 
		Choosing $2\Delta/\tau_0 = \delta'$ implies $\Pr[C] \geq 1 - \delta'$. 
		
		The first property now follows from the fact that 
		$\Pr[E^+_{A + \tau_0}(\widehat w, w) \cap C] \geq \Pr[E^+_{A + \tau_0}(\widehat w, w)] - \delta'$. Finally,  Bayes rule and the fact that 
		$\Pr[E^+_{A + \tau_0}(\widehat w, w)] \geq \eta_0$, implies $\Pr[C \mid E^+_{A + \tau_0}(\widehat w, w) ] \geq 1 - \delta'/ \Pr[E^+_{A + \tau_0}(\widehat w, w)] \geq 1- \delta'/\eta_0$. 
		
		The second property follows from the Bayes rule, 
		$\Pr[E^+_{A + \tau_0}(\widehat w, w)] \geq \eta_0$ and $\Pr[E^-_{A - \tau_0}(\widehat w, w)] \geq \eta_0$,  
		and the fact that $\Pr[\overline C] \leq \delta'$. 
		
		For $x \in E^+_{A + \tau_0}(\widehat w, w) \cap C$, 
		the third property follows from the fact that $\sgn(\cdot)$ is monotonically increasing and the fact that if $g(\widehat w \cdot x) - g(w \cdot x) - A > \tau_0$
		and $\abs{g( \widehat w \cdot x)- g(w^* \cdot x)} \leq \tau_0/2$, 
		then $g(w^* \cdot x) - g(w \cdot x) - A > \tau_0/2$. 
		A similar argument for the case when $x \in E^-_{A - \tau_0}(\widehat w, w) \cap C$ proves our result.  
		
	\end{proof}
	
	Let $R^+ := \{ y_i - g(w \cdot x_i) \mid x_i \in E^+_{A + \tau_0}(\widehat w, w)\}$ 
	and $R^- := \{ y_i - g(w \cdot x_i) \mid x_i \in  E^-_{A - \tau_0}(\widehat w, w) \}$ 
	for a specific choice of $\tau_0$.
	Our algorithm rejects $w$ if there is an $A$ such that \\
	$\abs{ \widehat \E[\sgn(r - A) \mid r \in R^+] - \widehat \E[\sgn(r - A) \mid r \in R^-]} \geq \alpha \min \{ \tau_0/8\sigma, 1/8 \}$. 
	An application of the properties from \Cref{claim:r_ppties} shows us that if Condition~\ref{cond1} 
	holds for $w$, then this is indeed the case for the true distribution.
	\begin{align*}
		&\abs{ \E[\sgn(g(w^* \cdot x) - g(w \cdot x) + \xi + \eps - A) \mid x \in E^+_{A + \tau_0}(\widehat w, w)] \\
			&\qquad - \E[\sgn(g(w^* \cdot x) - g(w \cdot x) + \xi + \eps - A) \mid x \in E^-_{A - \tau_0}(\widehat w, w)]} \\
		&=\abs{ \Pr[C \mid E^+_{A + \tau_0}(\widehat w, w)] ~\E[\sgn(g(w^* \cdot x) - g(w \cdot x) + \xi + \eps - A) \mid x \in E^+_{A + \tau_0}(\widehat w, w) \cap C] \\
			&\qquad + \Pr[\overline{C} \mid E^+_{A + \tau_0}(\widehat w, w)] ~\E[\sgn(g(w^* \cdot x) - g(w \cdot x) + \xi + \eps - A) \mid x \in E^+_{A + \tau_0}(\widehat w, w) \cap \overline{C}] \\
			&\qquad - \Pr[C \mid E^-_{A - \tau_0}(\widehat w, w)] ~\E[\sgn(g(w^* \cdot x) - g(w \cdot x) + \xi + \eps - A) \mid x \in E^-_{A - \tau_0}(\widehat w, w) \cap C] \\
			&\qquad - \Pr[\overline{C} \mid E^-_{A - \tau_0}(\widehat w, w)] ~\E[\sgn(g(w^* \cdot x) - g(w \cdot x) + \xi + \eps - A) \mid x \in E^-_{A - \tau_0}(\widehat w, w) \cap \overline{C}] } \\
		&\geq \abs{ (1- \delta'/\eta_0)~\E[\sgn(g(w^* \cdot x) - g(w \cdot x) + \xi + \eps - A) \mid x \in E^+_{A + \tau_0}(\widehat w, w) \cap C] \\
			&\qquad - ~\E[\sgn(g(w^* \cdot x) - g(w \cdot x) + \xi + \eps - A) \mid x \in E^-_{A - \tau_0}(\widehat w, w) \cap C]}  \\
		&\qquad -  2\delta'/\eta_0 \\
		&\geq \abs{ \E[\sgn(g(w^* \cdot x) - g(w \cdot x) + \xi + \eps - A) \mid x \in E^+_{A + \tau_0}(\widehat w, w) \cap C] \\
			&\qquad  -  \E[\sgn(g(w^* \cdot x) - g(w \cdot x) + \xi + \eps - A) \mid x \in E^-_{A - \tau_0}(\widehat w, w) \cap C]} - 3\delta'/\eta_0 \\
		&\geq \abs{ \E[\sgn(\xi + \eps + \tau_0/2) \mid x \in E^+_{A + \tau_0}(\widehat w, w) \cap C] \\
            &\qquad -  \E[\sgn(\xi + \eps - \tau_0/2) \mid x \in E^-_{A - \tau_0}(\widehat w, w) \cap C]} - 3\delta'/\eta_0 \\
		&= \abs{ \E[\sgn(\xi + \eps + \tau_0/2)] -  \E[\sgn(\xi + \eps - \tau_0/2)]} - 3\delta'/\eta_0 \\
		&= 2 \Pr[\abs{\xi + \eps} \leq \tau_0/2] -  3\delta'/\eta_0 \\
		&\geq 2 \alpha \Pr[\abs{\eps} \leq \tau_0/2] -  3\delta'/\eta_0 \\
		&\geq \alpha \min \{ \tau_0/4\sigma, 1/4\}.
	\end{align*}
	
	The final inequality follows by setting $\delta' <  \alpha \eta_0 \min \{ \tau/12\sigma, 1/12 \}) $, 
	and the fact that whenever $\tau_0 < 2$,  $\Pr[\abs{\eps} \leq \tau_0/2] \geq \min \{ \tau_0/2\sigma, 1/2 \}$.
	We will estimate 
	$\E[\sgn(g(w^* \cdot x) - g(w \cdot x) + \xi + \eps - A) \mid x \in E^+_{A + \tau_0}(\widehat w, w)] - \E[\sgn(g(w^* \cdot x) - g(w \cdot x) + \xi + \eps - A) \mid r \in E^-_{A - \tau_0}(\widehat w, w)]$ 
	upto an error of $\alpha \min \{ \tau_0/8\sigma,1/8 \}$.  
	For a fixed $\widehat w, w^*$ and $w$, this follows by estimating $\Pr[E^+_{A + \tau_0}(\widehat w, w)]$ and 
	$\E[\sgn(g(w^* \cdot x) - g(w \cdot x) + \xi + \eps - A) \mathbf{1}(x \in E^+_{A + \tau_0}(\widehat w, w))]$ (and the corresponding $E^-$ terms) each to an accuracy of 
	$\eta_0^2 \alpha \min \{ \tau_0/64\sigma,1/64 \}$. 
	Since both of these are expectations of random variables bounded by one, 
	Hoeffding's Lemma (\Cref{lem:hoeffding}) implies that $(64/\alpha^2 \eta_0^4 (\min \{ \tau_0/64\sigma,1/64 \})^2) \log(1/\delta)$ 
	samples suffice to achieve this approximation with a probability of $1-\delta$. 
	Let $\widehat \Pr$ and $\widehat \E$ denote the empirical expectation and probability respectively, 
	then an application of \Cref{fact:approx_q} to 
	$\Pr[E^+_{A + \tau_0}(\widehat w, w)]$, 
	$\E[\sgn(g(w^* \cdot x) - g(w \cdot x) + \xi + \eps - A) \mathbf{1}(x \in E^+_{A + \tau_0}(\widehat w, w))]$  
	and their respective empirical estimates implies $\abs{\widehat \E[\sgn(r - A) \mid r \in R^+] - \widehat \E[\sgn(r - A) \mid r \in R^+]} \leq 8 \eta_0^2 \alpha \min \{ \tau_0/64\sigma,1/64 \} / \eta_0^2 \leq \alpha \min \{ \tau_0/8\sigma,1/8 \}$. 
	
	A union bound over all possible $\cW$ candidates for $w$ and $\widehat w$ tells us 
	that a sample complexity of $(64/\alpha^2 \eta_0^4 (\min \{ \tau_0/64\sigma,1/64 \})^2) \log(\abs{\cW}^2/\delta)$ suffices.
	
	Suppose $w$ is not rejected, then we know that Condition~\ref{cond2} holds. 
	Another application of Markov's inequality similar to before gives us $\Pr\brac{\abs{g(\widehat w \cdot x)- g(w^* \cdot x)} \geq \tau_0/2} \leq 2\Delta/\tau_0 = \delta'$. 
	Any $x$ satisfying $\abs{g(\widehat w \cdot x)- g(w^* \cdot x)} \leq \tau_0$ and $g(w^* \cdot x) - g(w \cdot x)  - A > 2\tau_0$ must also satisfy $g(\widehat w \cdot x) - g(w \cdot x) - A > \tau_0$. 
	This implies that $\Pr[E^+_{A + 2\tau_0}(w^*, w)]  \leq \eta_0 + \delta'$. 
	A similar argument shows that $\Pr[E^-_{A - 2\tau_0}(w^*, w)]  \leq \eta_0 + \delta'$. 
	By choosing $2 \tau_0 = \tau$ and $\eta_0 + \delta' = \tau / 2$, every hypothesis we return satisfies
	$\Pr[E^-_{A - \tau}(w^*, w)]  \leq \tau / 2$ and $\Pr[E^+_{A +\tau}(w^*, w)]  \leq \tau / 2$. 
	The constraints on the variables are satisfied when $\eta_0 = \delta' = \tau / 4$ and $\tau_0 = \tau / 2$, which amounts to $\Delta < \tau^2/16$.

	If $(\Delta, \tau)$-identifiability holds, every element $w$ in the set of candidates that remains satisfies $\E_x \Brac{\abs{g(w^* \cdot x) - g(w \cdot x)}} \leq \Delta$. To check that this is the case, the algorithm tests if every pair of candidates $u, v$ in $\cW$ is at most $3\Delta$-close, i.e. $\widehat \E_x \Brac{\abs{g(u \cdot x) - g(v \cdot x)}} \leq 3\Delta$. If this is the case, we return any candidate in the set. Otherwise we get a polynomial sized list $ \cL$ with $\widehat w \in \cL$.

\end{proof} 

\section{Proof of Main Theorem}
\label{sec:Proof of Main Theorem}

Here we state and prove our main theorem, 
which is a more detailed version of \Cref{thm:main_informal}. 

\begin{theorem}[Main Result] \label{thm:main-detailed}
We first define a few variables and their relationships to $\Delta$ (the desired final accuracy), $\alpha$ (the probability of being an inlier), $R$ (an upper bound on $\| w^*\|$)
and $\sigma$ (the standard deviation of the additive Gaussian noise). 

Let
$\Delta' = \min(\Delta, \tau^2/16)$. 
$\gamma = \min(\Delta/4\sigma, 1/2)$,
$T \gtrsim (R/\gamma \alpha)^2$,
$m_1 \gtrsim R^2 \ln(T/\delta)/(\gamma \alpha \Delta)^2$
and $W \gtrsim T(\gamma \alpha \Delta/64 R)$.

There is an algorithm, which, given $\Delta, \alpha, R$ and $\sigma$
runs in time $O(dTm_1)$, 
draws $m_1 \gtrsim \alpha^{-2} \log(R/\Delta \alpha \delta) \Paren{ R^2\sigma^2 / \Delta^4}$ samples 
from $\text{GLM-Ob}(g, \sigma, w^*)$ and 
returns a $T(\gamma \alpha \Delta/64 R)$-sized list of candidates, 
one of which achieves excess loss at most $\Delta$. 

Moreover, if the instance is $(\Delta, \tau)$-identifiable
then, there is an algorithm which takes the parameters 
$\Delta, \alpha, \sigma, R$ and $\tau' \leq \tau$,
draws $$m \gtrsim \alpha^{-2} \log(W/\delta) \Paren{ R^2 \sigma^2/ ( \Delta'^4 + 1/(\tau' \min (\tau'/\sigma,1 ))^{2} }$$ samples from $\text{GLM-Ob}(g, \sigma, w^*)$, runs in time  $O(d m W^2)$ and returns a single candidate. 
\end{theorem}
\begin{proof}
	Recall that $P$ is a uniform partition of $[-1, 1]$ with granularity $p = \gamma \alpha \Delta/64 R$.
	For each $c \in P$ we run the algorithm from \Cref{thm:OCO_reduction} for $T \gtrsim (R/\gamma \alpha)^2$ steps where $\gamma = \sigma/4\Delta$.
	From the lemma, we know that when $\abs{c - \E_{\xi ,\eps} [\sgn(\xi + \eps)]} \leq p$ 
	one of the candidates generated by the online gradient descent algorithm satisfies $\E [ \abs{ g(w^* \cdot x) - g(\widehat w \cdot x) } ] \leq \Delta$.  

 For the second part of the theorem, we set $\Delta' = \min(\Delta, \tau^2/16)$ and run the OCO algorithm above to get a larger list of candidates, one of which achieves excess loss $\Delta'$. 
	Finally, we collect all $T/p \lesssim (1/\Delta')~(R/\gamma \alpha)^3 = |\cW|$ candidates and run our pruning algorithm \Cref{alg:approx_search} on them. 
	Then \Cref{lem:prune_bad_candidates} returns a list satisfying our final guarentee. 
	
	Putting together the sample complexities of the lemmas, we see that for this second part,
	\[m \gtrsim \log(\abs{\cW}^2/\delta)/(\alpha \tau (\min (\tau/2\sigma,1 )))^2  + R^2 \ln(T/\delta) / (\gamma \alpha \Delta')^2 \;. \]
\end{proof}

\end{document}